\documentclass[smallextended]{svjour3}       
\smartqed  
\usepackage{graphicx}

\usepackage{wrapfig}

\usepackage{graphicx}
\usepackage{graphics}
\usepackage{epsfig}
\usepackage{epstopdf}
\usepackage{amsmath}
\usepackage{latexsym}
\usepackage{subfig}
\usepackage{wrapfig}
\usepackage{multirow}
\usepackage{amsfonts}
\usepackage{cite}
\usepackage{amssymb}
\usepackage{caption}
\usepackage{graphicx}
\usepackage{graphics}
\usepackage{epsfig}
\usepackage{epstopdf}
\usepackage{amsmath}
\usepackage{latexsym}
\usepackage{wrapfig}
\usepackage{amsfonts}
\usepackage{cite}
\usepackage{ifpdf}
\usepackage{algorithm}
\usepackage{algorithmic}
\usepackage{color}

\makeatletter

\newcommand{\Rmnum}[1]{\expandafter\@slowromancap\romannumeral #1@}
\makeatother

\newcommand{\deleted}[1]{}

\begin{document}

\title{A Unified Framework for Clustering Constrained Data without Locality Property\thanks{This work was supported in part by NSF through  grants IIS-1115220, IIS-1422591, CCF-1422324, CNS-1547167, CCF-1656905, and CCF-1716400. The first author was also supported by a start-up fund from Michigan State University. A preliminary version of this paper has appeared in Twenty-Sixth Annual ACM-SIAM Symposium on Discrete Algorithms (SODA 2015)\cite{DX15}.}}


\author{Hu Ding        \and
        Jinhui Xu 
}


\institute{Hu Ding \at
            Department of Computer Science and Engineering\\
           Michigan State University\\
           School of Computer Science and Technology\\
University of Science and Technology of China\\           
              \email{huding@msu.edu, huding@ustc.edu.cn}           
           \and
           Jinhui Xu \at
           Department of Computer Science and Engineering\\
State University of New York at Buffalo\\
 \email{jinhui@buffalo.edu}}

\date{Received: date / Accepted: date}

\maketitle

\begin{abstract}
In this paper, we consider a  class of constrained clustering problems of points in $\mathbb{R}^{d}$, where $d$ could be rather high. A common feature of these problems is that their optimal clusterings no longer have the locality property (due to the additional constraints), which is a key property required by many algorithms for their unconstrained counterparts.  To overcome the difficulty caused by the loss of locality, we present in this paper a unified framework, called {\em Peeling-and-Enclosing (PnE)}, to iteratively solve two variants of the constrained clustering problems, {\em constrained $k$-means clustering} ($k$-CMeans) and {\em constrained $k$-median clustering} ($k$-CMedian). Our framework generalizes Kumar {\em et al.}'s elegant $k$-means clustering approach \cite{KSS} from unconstrained data to constrained data, and is based on two standalone geometric techniques, called {\em Simplex Lemma} and {\em Weaker Simplex Lemma}, for $k$-CMeans and $k$-CMedian, respectively. The simplex lemma (or weaker simplex lemma) enables us to efficiently approximate the mean  (or median) point of an unknown set of  points  by searching a small-size grid, independent of the dimensionality of the space,  in a simplex (or the surrounding region of a simplex), and thus can be used to handle high dimensional data. If $k$ and $\frac{1}{\epsilon}$ are fixed numbers, our framework generates,  in nearly linear time ({\em i.e.,} $O(n(\log n)^{k+1}d)$), $O((\log n)^{k})$ $k$-tuple candidates for the $k$ mean or median points, and one of them induces a $(1+\epsilon)$-approximation for  $k$-CMeans or $k$-CMedian, where $n$ is the number of points. Combining this unified framework with a problem-specific selection algorithm (which determines the best $k$-tuple candidate), 
we obtain a $(1+\epsilon)$-approximation for each of the constrained clustering problems. Our framework  improves considerably the best known results for these problems. We expect that our technique will be applicable to  other constrained clustering problems without locality.


\end{abstract}

%
%

\section{Introduction}
\label{sec-intro}

Clustering is one of the most fundamental problems in computer science, and finds numerous applications in many different areas, such as data management, machine learning, bioinformatics, networking, {\em etc.} \cite{JMF}. The common goal of many clustering problems is to partition a set of given data items into a number of clusters so as to minimize the total cost measured by a certain objective function.  For example, the popular {\em $k$-means} (or {\em $k$-median}) clustering seeks $k$ mean (or median) points to induce a partition of the given data items so that the average squared distance (or the average distance) from each data item to its closest mean (or median) point is minimized.
Most existing clustering techniques assume that the data items are independent from each other and therefore can ``freely'' determine their memberships in the resulting clusters ({\em i.e.,}  a data item does not need to pay attention to the clustering of others). 
However, in many real-world applications, data items are often constrained or correlated, which require a great deal of effort to handle such additional constraints.  
In recent years, considerable attention has been paid to various types of constrained clustering problems  and a number of techniques, such as  {\em $l$-diversity clustering} \cite{LYZ}, {\em $r$-gather clustering} \cite{APF,EHR,HR13},  {\em capacitated clustering} \cite{ABS13,CHK,KS00},  {\em chromatic clustering} \cite{DX11,ADH}, and {\em probabilistic clustering} \cite{GM09,CM08,LSS}, have been obtained. 
In this paper, we study a class of constrained clustering problems  of points in Euclidean space.  


{\em Given a set of points $P$ in $\mathbb{R}^d$, a positive integer $k$, and a constraint $\mathbb{C}$, the constrained $k$-means (or $k$-median) clustering problem is to partition $P$ into $k$ clusters so as to  minimize the objective function of the ordinary $k$-means (or $k$-median) clustering and satisfy the constraint $\mathbb{C}$. In general, the problems are denoted by $k$-CMeans and  $k$-CMedian, respectively.}
%

The detailed definition for each individual problem is given in Section~\ref{sec-application}. 
Roughly speaking, data constraints can be imposed at either {\em cluster} or {\em item} level. Cluster level constraints are restrictions on the resulting clusters, such as the size of the clusters \cite{APF} or their mutual differences \cite{ZLM}, while item level constraints  are mainly on data items inside each cluster, such as the coloring constraint  which prohibits  items of the same color being clustered into one cluster  \cite{ADH,DX11,LYZ}.  

\begin{figure}
  \begin{center}
  \subfloat[]{\label{fig-voronoi1}\includegraphics[width=0.3\textwidth]{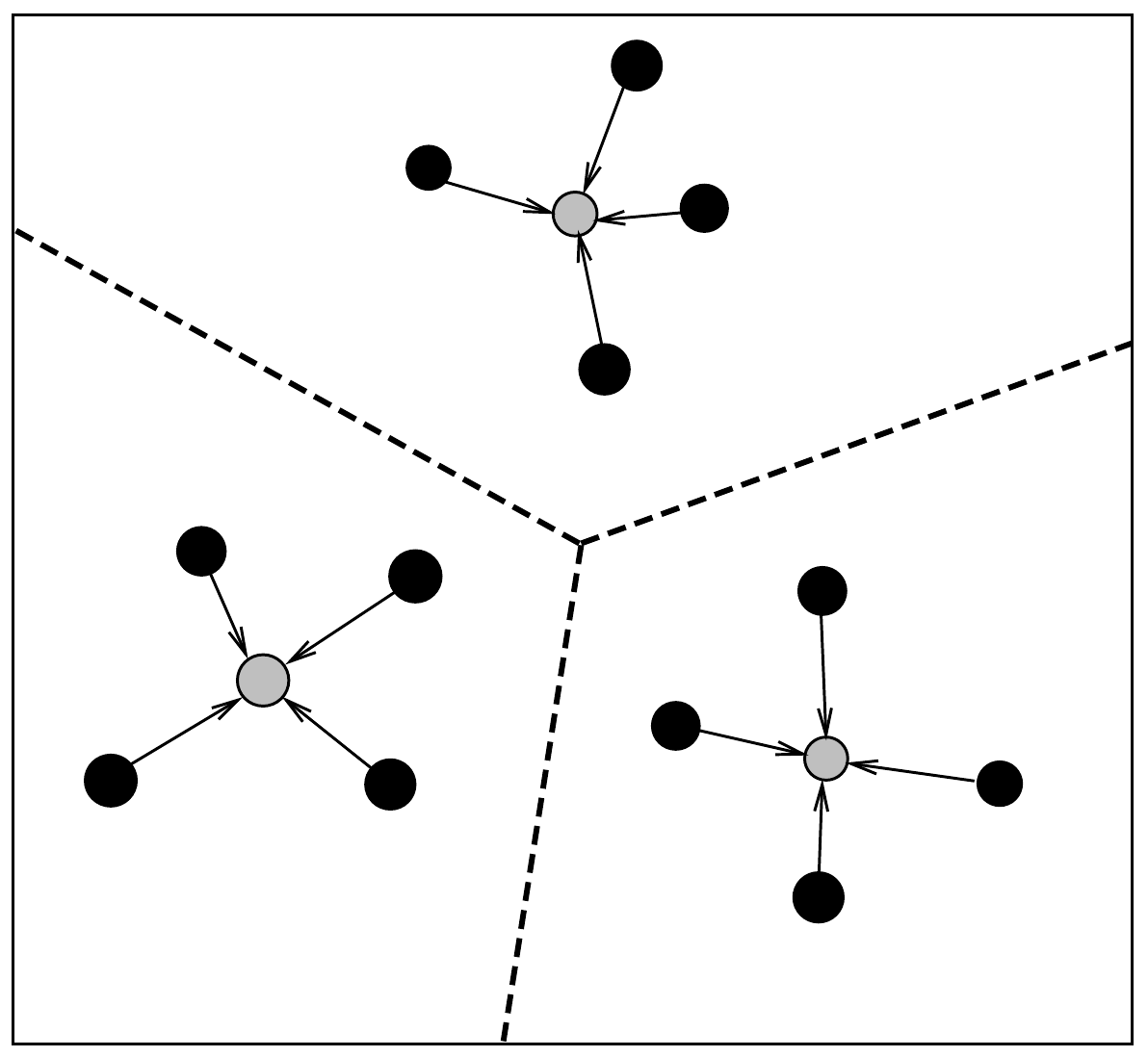}}
  \hspace{1in}
  \subfloat[]{\label{fig-voronoi2}\includegraphics[width=0.3\textwidth]{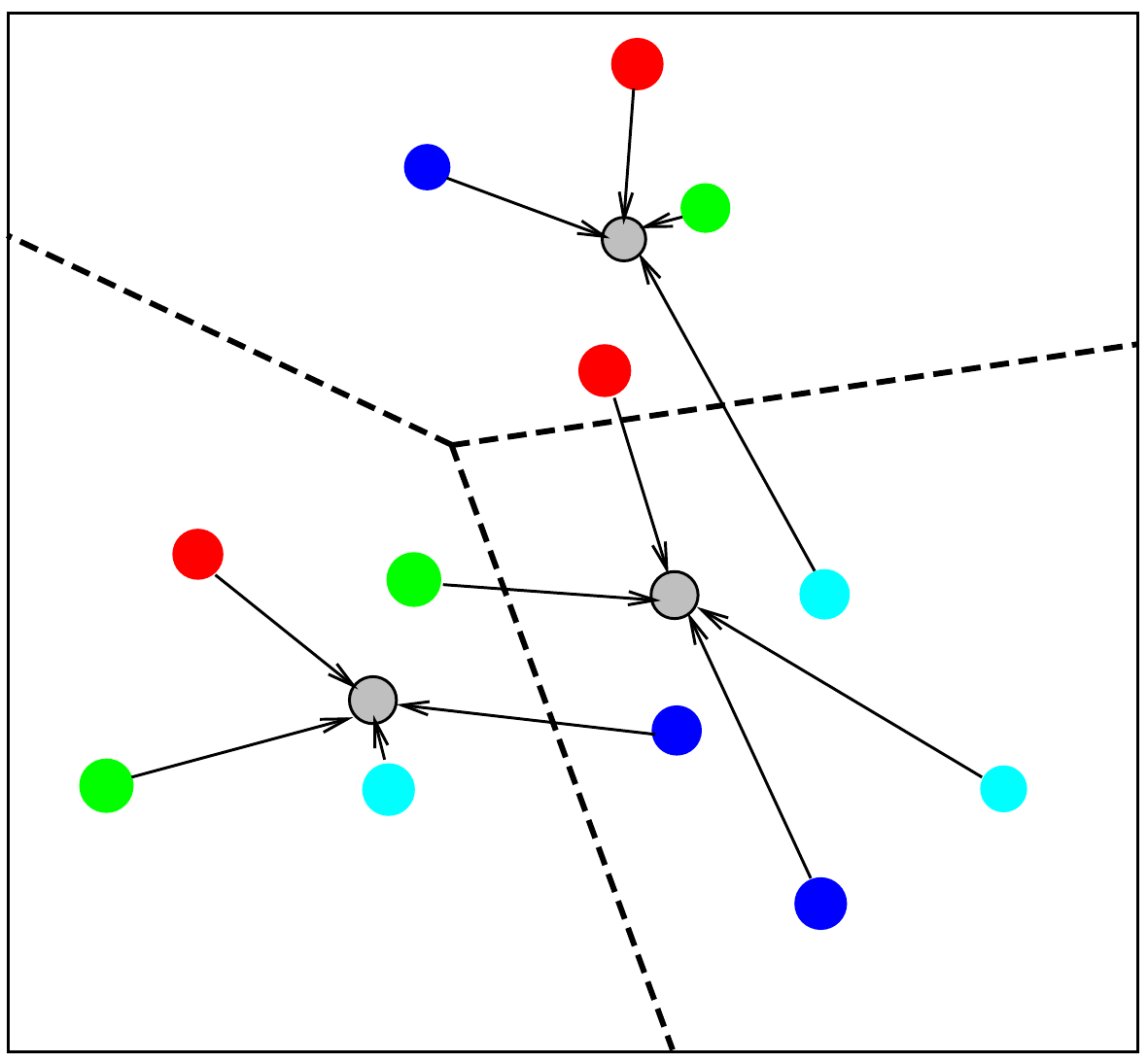}}
   \end{center}
  \caption{(a) The Voronoi diagram induced by the mean points ({\em i.e.,} the grey points) of $k$-means clustering for $k=3$; (b) the Voronoi diagram induced by the mean points of chromatic $k$-means clustering, where the points sharing the same color should be in different clusters.}
\end{figure}

The additional constraints could considerably change the nature of the clustering problems. For instance, one key property exhibited in many unconstrained clustering problems is the so called {\em locality} property, which  indicates that each cluster is located entirely inside the Voronoi cell  of its center  ({\em e.g.,} the mean, median, or center point) in the Voronoi diagram of all the centers \cite{IKI} (see Figure~\ref{fig-voronoi1}). Existing algorithms for these clustering problems often  rely on such a property \cite{KSS,BHI,AV07,C09,FKK,IKI,M00,OSS}. However, due to the additional constraints, the locality property may no longer exist (see Figure~\ref{fig-voronoi2}). Therefore, we need new techniques to overcome  this challenge. 

\subsection{Our Main Results}
\label{sec-mainresult}
In this paper we present a unified framework called {\em Peeling-and-Enclosing (PnE)}, based on two standalone geometric techniques called {\em Simplex Lemma} and {\em Weaker Simplex Lemma},  to solve a class of constrained clustering problems  without the locality property in Euclidean space, where the dimensionality of the space could be rather high and the number $k$ of clusters is assumed to be some fixed number. Particularly, we investigate the constrained $k$-means ($k$-CMeans) and $k$-median ($k$-CMedian) versions of these problems.
For the $k$-CMeans problem, our unified framework generates in $O(n(\log n)^{k+1}d)$ time a set of $k$-tuple candidates of cardinality $O((\log n)^{k})$  for the to-be-determined $k$ mean points.  We show that among the set of candidates,  one of them  induces a $(1+\epsilon)$-approximation for $k$-CMeans. To find out the best $k$-tuple candidate, a problem-specific selection algorithm is needed for each individual constrained clustering problem (note that due to the additional constraints, the selection problems may not be trivial).
Combining the unified framework with the selection algorithms, we obtain a $(1+\epsilon)$-approximation for each constrained clustering problem in the considered class. Our results considerably improve (in various ways) the best known algorithms for all these problems (see the table in Section~\ref{sec-related}). 
Our techniques can also be extended to $k$-CMedian to achieve $(1+\epsilon)$-approximations for these problems with the same time complexities. Below is a list of the constrained clustering problems considered in this paper. We expect that our technique will be applicable to other clustering problems without locality property, as long as the corresponding selection problems  can be solved.

\begin{enumerate}

\item \textbf{$l$-Diversity Clustering.} In this problem, each input point is associated with a color, and each cluster has no more than a fraction  $\frac{1}{l}$ (for some constant $l>1$) of its points sharing the same color.
The problem is motivated by a widely-used privacy preserving model called {\em $l$-diversity}  \cite{MGK,LYZ}  in data management, which requires that each block contains no more than a fraction $\frac{1}{l}$ of items sharing the same sensitive attribute.  

\item \textbf{Chromatic Clustering.} In \cite{DX11}, Ding and Xu  introduced a new clustering problem called {\em chromatic clustering}, which requires that the points with the same color should be clustered in different clusters.  It is motivated by a biological application for clustering  chromosome homologs in a population of cells, where  homologs from the same cell should be clustered into different clusters. 
Similar problem also appears in applications related to transportation system design \cite{ADH}. 

\item \textbf{Fault Tolerant Clustering.}
The problem of {\em fault tolerant clustering} assigns each point $p$ to its $l$ nearest cluster centers for some $l\ge 1$, and counts all the $l$ distances 
as its cost.  The problem has been extensively studied in various applications for achieving better fault tolerance~\cite{CGR,KPS,SS03,KR13,HHL}.


\item \textbf{$r$-Gather Clustering.} This clustering problem requires each of the clusters to contain at least $r$ points for some  $r>1$.   
It is motivated by the {\em $k$-anonymity} model for privacy preserving~\cite{S02,APF}, where each block contains at least $k$ items \footnote{We use $r$ here, instead of $k$, since $k$ often denotes the number of clusters in a clustering problem.}.

\item \textbf{Capacitated Clustering.} This clustering problem has an upper bound on the size of each cluster, and 
finds various applications in data mining and resource assignment \cite{KS00,CHK}.

\item \textbf{Semi-Supervised Clustering.} Many existing clustering techniques, such as {\em ensemble clustering}~\cite{SG02,Sin10} and {\em consensus clustering}~\cite{ACN,CW10}, make use of a priori knowledge. Since such clusterings are not always based on the geometric cost  ({\em e.g.,} $k$-means cost) of the input, thus a more accurate way of clustering is to consider both the priori knowledge and  the geometric cost. We consider the following {\em semi-supervised clustering} problem: 
given a set $P$ of points and a clustering $\overline{\mathcal{S}}$ of $P$ (based on the priori knowledge),  partition $P$ into $k$ clusters so as to minimize the sum (or some function) of the geometric cost and the difference with the given clustering $\overline{\mathcal{S}}$.  Another related problem is {\em evolutionary clustering}~\cite{CKT}, where the clustering  in each time point needs to minimize not only the geometric cost, but also the total shifting  from the clustering in the previous time point (which can be viewed as $\overline{\mathcal{S}}$). 

\item \textbf{Uncertain Data Clustering.} Due to the unavoidable error, data for clustering are not always precise.  
This motivates us to consider the following 
{\em probabilistic clustering} problem \cite{GM09,CM08,LSS} : given a set of ``nodes'' with each represented as a probabilistic distribution over a point set in $\mathbb{R}^d$, 
group the nodes into $k$ clusters so as to  minimize the expected cost with respect to the probabilistic distributions. 


\end{enumerate}

\noindent\textbf{Note:} Following our work published in \cite{DX15}, Bhattacharya {\em et al.}~\cite{BJK} improved the running time for finding the candidates of $k$-cluster centers from nearly linear to linear based on the elegant $D^2$-sampling. Their work also follows the framework of clustering constrained data, i.e., generating the candidates and selecting the best one by a problem-specific selection algorithm, presented in this paper. Our paper represents the first systematically theoretical study of the constrained clustering problems. Some of the underlying techniques, such as  Simplex Lemma and Weaker Simplex Lemma, are interesting in their rights, which have already been used to solve other problems \cite{DGX} ({\em e.g.,} the popular ``truth discovery" problem in data mining). 

%
%
%

 \subsection{Related Works} 
\label{sec-related}

\begin{center}
\begin{table*}[ht]
{\small
\hfill{}
\begin{tabular}{ |c|p{7.5cm}| }
\hline
  Problems & \hspace{3cm} Existing Results  \\ \hline
$l$-diversity clustering & $2$-approx. for metric $k$-centers \cite{LYZ} (only for a restricted version of $l$-diversity clustering)\\ 
\hline
  chromatic clustering & $(1+\epsilon)$-approx. for chromatic $k$-cones clustering in $\mathbb{R}^d$ \cite{DX11}; $(1+\epsilon)$-approx. for $2$-center in $\mathbb{R}^2$ \cite{ADH} \\ \hline
  fault tolerant clustering& $4$ and $93$-approx. for uniform and non-uniform metric $k$-median \cite{SS03,HHL}; $2$-approx. for metric $k$-centers \cite{CGR,KPS}; \\ \hline
  $r$-gather  clustering & $2$-approx. for metric $k$-centers and $4$-approx. for metric $k$-cellulars \cite{APF}; $(4+\epsilon)$-approx. for $k$-centers in  constant dimensional space \cite{EHR,HR13}\\\hline
  capacitated clustering & $6$ and $7$-approx. for metric $k$-centers with uniform and non-uniform capacities \cite{KS00,CHK}  \\\hline
  semi-supervised clustering & Heuristic algorithms \cite{BBM,GTC,WCS,WCS00} \\ \hline
  uncertain data clustering &  $(1+\epsilon)$-approx. for $k$-means and unassigned $k$-median \cite{CM08}; $(3+\epsilon)$-approx. for assigned $k$-median \cite{CM08,XX10}; $(1+\epsilon)$-approx. for assigned $k$-median in constant dimensional space \cite{LSS}; $O(1)$-approx. for $k$-centers \cite{GM09}\\ \hline \hline
\multicolumn{2}{ |p{10.5cm}| }{Our results: $(1+\epsilon)$-approx. of $k$-means and $k$-median for all the $7$ problems in $\mathbb{R}^d$ where $d$ could be rather high.} \\
\hline
 \end{tabular}
}
\hfill{}
\caption{\small{Existing and our new results for the class of constrained clustering problems.}}
\label{tb-result}
\end{table*}
\end{center}
 \normalsize
 
The above $7$ constrained clustering problems have been extensively studied in the past and a number of theoretical results have been obtained (in addition to many heuristic/practical solutions). Table~\ref{tb-result} lists the best known theoretical results for each of them. It is clear that most existing results are either constant approximations or only for some restricted versions ({\em e.g.,} constant dimensional space, etc.), and therefore can be improved by our techniques.

For the related traditional Euclidean $k$-means and $k$-median clustering problems, extensive research has been done in the past. 
 Inaba {\em et al.} \cite{IKI} showed that an exact $k$-means clustering can be computed in $O(n^{O(dk)})$ time for $n$ points in $\mathbb{R}^{d}$.
Arthur and Vassilvitskii\cite{AV07} presented the $k$-means++ algorithm that achieves the expected $O(\log k)$ approximation ratio. 
Ostrovsky {\em et al.} \cite{OSS} provided a $(1+\epsilon)$-approximation for well-separated points. Based on the concept of stability, Awasthi {\em et al.} \cite{ABS10} presented the PTAS for the problems of $k$-means and $k$-median clustering.  Matousek~\cite{M00} obtained a nearly linear time $(1+\epsilon)$-approximation for any fixed $d$ and $k$. Similar result  for $k$-median has also been achieved by Kolliopoulos and Rao \cite{KR07}. 
 Later, Fernandez de la Vega {\em et al.} \cite{FKK} and  Bad\u{o}iu {\em et al.} \cite{BHI} achieved nearly linear time $(1+\epsilon)$-approximations for high dimensional $k$-means and $k$-median clustering, respectively, for fixed $k$. Kumar {\em et al.} \cite{KSS} showed that  linear-time  randomized  $(1+\epsilon)$-approximation algorithms can be obtained  for several Euclidean clustering problems (such as $k$-means and $k$-median) in any dimensional space. Recently, this technique
  has been further extended to several  clustering problems with non-metric distance functions  \cite{ABS}.
 Later, Jaiswal {\em et al.}~\cite{JKS} applied a non-uniform sampling technique, which is called $D^2$-sampling, to simplify and improve the result in \cite{KSS}; their algorithm can also handle the non-metric distance clustering problems studied in \cite{ABS}.
Using the {\em core-set} technique,  a series of improvements have been achieved  for high dimensional clustering problems  \cite{FL11}.

As for the hardness of the problem, Dasgupta \cite{D08} showed that it is NP-hard for $k$-means clustering in high dimensional space even if $k=2$; Awasthi et al.~\cite{ACK15} proved that there is no PTAS for $k$-means clustering if both $d$ and $k$ are large, unless $P=NP$. Guruswami  and Indyk \cite{GI03} showed that it is NP-hard to obtain any PTAS for $k$-median clustering if $k$ is not a constant and $d$ is $\Omega(\log n)$.

Besides the traditional  clustering models, 
Balcan  {\em et al.}  considered  the problem of  finding the clustering with small difference from the unknown ground truth~\cite{BBG,BB09}. 

\subsection{Our Main Ideas}
\label{sec-idea}
Most existing $k$-means or $k$-median clustering algorithms in Euclidean space consist of two main steps:  (1) identify the set of $k$ mean or median points and (2) partition the input points into $k$ clusters based on these mean or median points (we call this step \textbf{Partition}). 
\textbf{Note that for some constrained clustering problems, the Partition step may not be trivial.}
More formally, we have the following definition.

\begin{definition}[Partition Step]
\label{def-partition}
Given an instance $P$ of $k$-CMeans (or $k$-CMedian) and $k$ cluster centers (i.e., the mean or median points), the Partition step is to form $k$ clusters of $P$, where the clusters should satisfy the constraint and each cluster is assigned to an individual cluster center, such that the objective function of the ordinary $k$-means (or $k$-median) clustering is minimized.
\end{definition}



To determine the set of $k$ mean or median points in step (1), most existing algorithms (either explicitly or implicitly) rely on the locality property. To shed some light on this, consider a representative and elegant approach by Kumar {\em et al.} \cite{KSS} for $k$-means clustering. Let $\{Opt_1,\cdots,Opt_k \}$ be the set of $k$ unknown optimal clusters in non-increasing order of their sizes. Their approach uses random sampling and sphere peeling to iteratively find  $k$ mean points. At the $j$-th iterative step, it draws $j$-$1$ peeling spheres centered at the $j$-$1$ already obtained  mean points, and takes a random sample on the 
points outside the peeling spheres to find the $j$-th mean point. 
Due to the locality property, the points belonging to the first $j$-$1$ clusters lie inside their corresponding $j$-$1$ Voronoi cells; that is, for each peeling sphere, most of the covered points belong to their corresponding cluster, and thus ensures the correctness of the peeling step.

However, when the additional constraint (such as coloring or size) is imposed on the points, the locality property may no longer exist (see Figure \ref{fig-voronoi2}), and thus the correctness of the peeling step cannot always be guaranteed.   
In this scenario, the core-set technique~\cite{FL11} is also unlikely to be able to resolve the issue. The main reason is that although the core-set can greatly reduce the size of the input points, it is quite challenging to impose the constraint through the  core-set.  


To overcome this challenge, we present a unified framework, called {\em Peeling-and-Enclosing (PnE)}, in this paper, based on a standalone new geometric technique called {\em Simplex Lemma}. The simplex lemma aims to address the major obstacle encountered by the peeling strategy in \cite{KSS} for constrained clustering problems. More specifically,  due to the loss of locality, at the $j$-th peeling step, the points of the $j$-th cluster $Opt_{j}$ could be scattered  over all the Voronoi cells of the first $j$-$1$ mean points, and therefore their mean point  can no longer be simply determined by the sample outside the $j$-$1$ peeling spheres.
To resolve this issue, our main idea is to view $Opt_{j}$ as the union of  $j$ unknown subsets, $Q_{1}, \cdots, Q_{j}$, 
 with each $Q_{l}$, $1\le l \le j$-$1$, being 
the set of points 
inside the  Voronoi cell (or peeling sphere) of the obtained $l$-th mean point  and $Q_{j}$ being the set of remaining points of $Opt_{j}$.  After approximating the mean point of each  unknown subset by using random sampling, we build a simplex  to enclose a region which contains the mean point of $Opt_{j}$, and then search  the simplex region for  a good approximation of the $j$-th mean point.
To make this approach work, we need to overcome two 
difficulties: \textbf{(a) how to generate the desired simplex to contain the $j$-th mean point}, and \textbf{(b) how to efficiently search the (approximate) $j$-th mean point  inside the simplex}.

For difficulty (a), our idea is to use  the already determined  $j$-$1$ mean points (which can be shown that they are also the approximate mean points of $Q_{1}, \cdots, Q_{j-1}$, respectively)  and another point, which is the mean of those points in $Opt_{j}$ outside the peeling spheres (or Voronoi cells) of the first $j$-$1$ mean points ({\em i.e.,} $Q_{j}$), to build a ($j$-$1$)-dimensional simplex to contain the $j$-th mean point.  Since we do not know how $Opt_{j}$ is partitioned ({\em i.e.,} how $Opt_{j}$ intersects the $j$-$1$ peeling spheres), we vary the radii of the peeling spheres $O(\log n)$ times to guess the partition and generate a set of simplexes, where the radius candidates are based on an upper bound of the optimal value determined by a novel estimation algorithm (in Section \ref{sec-upper}). We show that among the set of simplexes, one of them contains the $j$-th (approximate) mean point.

For difficulty (b), our simplex lemma (in Section \ref{sec-simplex}) shows that if each vertex $v_{l}$ of the simplex $\mathcal{V}$ is the (approximate) mean point of $Q_{l}$, then we can find a good approximation of the mean point of $Opt_j$ by searching a small-size grid inside $\mathcal{V}$. A nice feature of the simplex lemma is that the grid size is independent of the dimensionality of the space and thus can be used to handle high dimensional data. In some sense, our simplex lemma can be viewed as a considerable generalization of the well-known sampling lemma ({\em i.e.,} Lemma~\ref{lem-dis} in this paper) in \cite{IKI}, which has been widely used  for estimating the mean of a point set through random sampling~\cite{FMS,IKI,KSS}. Different from Lemma~\ref{lem-dis}, which requires a global view of the point set (meaning that the sample needs to be taken from the point set),  our simplex lemma only requires some partial views ({\em e.g.,} sample sets are taken from those unknown subsets whose size might be quite small). 
If $Opt_j$ is the point set, our simplex lemma enables us to bound the error by the variance\footnote{Given a point set in Euclidean space, its ``variance'' is the average of the squared distances from the points to their mean point. 
} of $Opt_j$ ({\em i.e.,} a local measure)  and the optimal value of the clustering problem on the whole instance $P$ ({\em i.e.,} a global measure), 
and thus helps us to ensure the quality of our solution.

For the $k$-CMedian problem, we show that although the simplex lemma no longer holds since the median point may lie outside the simplex, a weaker version (in Section~\ref{sec-wsl}) exists, which searches a surrounding region of the simplex. Thus our Peeling-and-Enclosing framework works for both $k$-CMeans and $k$-CMedian. It  generates in total 
$O((\log n)^{k})$ $k$-tuple candidates for the constrained $k$ mean or median points. To  determine the best $k$ mean or median points, we need to use the property of each individual problem to design a selection algorithm. The selection algorithm takes each $k$-tuple candidate, computes 
a clustering ({\em i.e.,} completing the Partition  step) satisfying the additional constraint, and outputs the $k$-tuple with the minimum cost. 
We present a selection algorithm for each considered problem in Sections~\ref{sec-application} and \ref{sec-application2}. 


%


%

\section{Simplex Lemma}
\label{sec-simplex}

\begin{figure} 
  \begin{center}
  \subfloat[]{\label{fig-simplex}\includegraphics[width=0.35\textwidth]{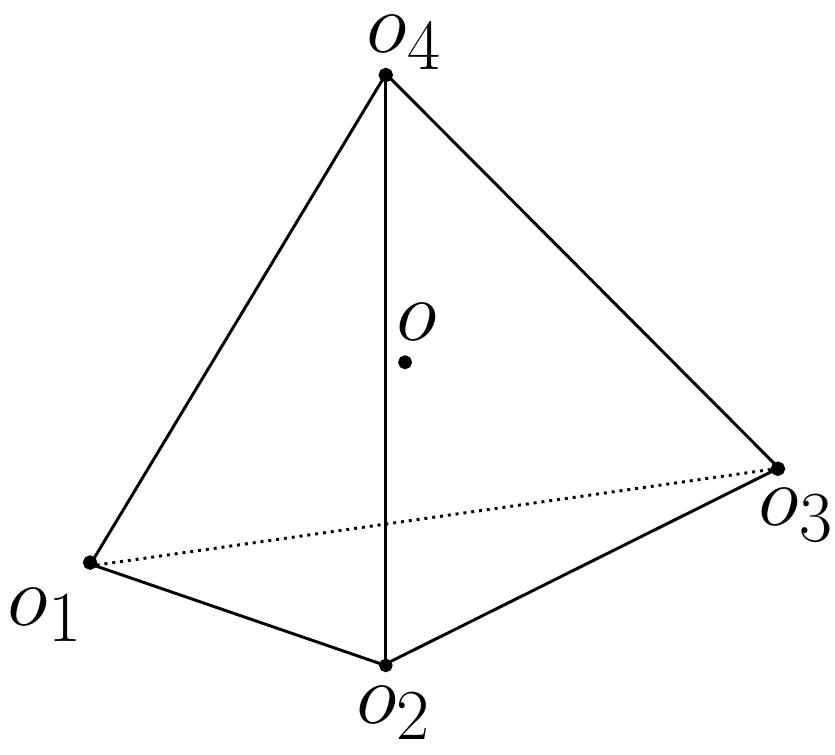}}
 \hspace{0.5in}
  \subfloat[]{\label{fig-simplex2}\includegraphics[width=0.35\textwidth]{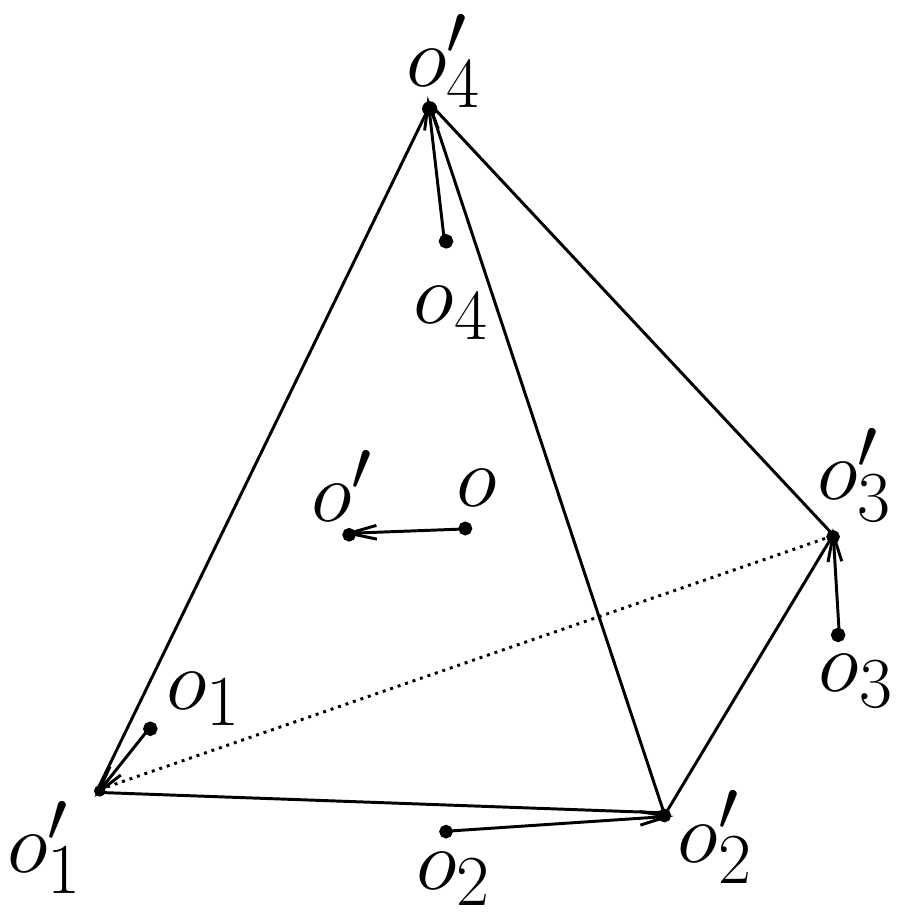}}
   \end{center}
  \caption{Examples for Lemma \ref{lem-simplex} and Lemma \ref{lem-shift} with $j=4$ respectively.}
\end{figure}

In this section, we present the {\em Simplex Lemma} for approximating the mean point of an \textbf{unknown} point set $Q$,  where the only known information is a set of $j$ points with each of them being an approximate mean point of an unknown subset of $Q$. In Section~\ref{sec-wsl}, we show how to extend the idea to approximate median point by the {\em Weaker Simplex Lemma}. The two lemmas are keys to solving the $k$-CMeans and $k$-CMedian problems.

\begin{lemma}[Simplex Lemma \Rmnum{1}]
\label{lem-simplex}
Let $Q$  be a set of points in $\mathbb{R}^d$ with a partition of $Q=\cup^j_{l=1} Q_l$ and $Q_{l_1}\cap Q_{l_2}=\emptyset$ for any $l_1\neq l_2$. 
Let $o$ be the mean point of $Q$, and $o_l$ be the mean point of $Q_l$ for $1\leq l\leq j$. Let the variance of $Q$ be $\delta^2=\frac{1}{|Q|}\sum_{q\in Q}||q-o||^2$,  and $\mathcal{V}$ be the simplex determined by $\{o_1, \cdots, o_j\}$. 
Then for any $0<\epsilon\leq 1$, it is possible to construct a grid of size $O((8j/\epsilon)^j)$ inside $\mathcal{V}$ such that at least one grid point $\tau$ satisfies the inequality $||\tau-o||\leq\sqrt{\epsilon}\delta$.
\end{lemma}


Figure \ref{fig-simplex} gives an example for Lemma \ref{lem-simplex}. To prove Lemma \ref{lem-simplex}, we first introduce the following  lemma.

 \begin{lemma}
\label{lem-close}
Let $Q$ be a set of points in $\mathbb{R}^d$, and $Q_{1}$ be a subset containing $\alpha |Q|$ points for some $0<\alpha\leq 1$. Let $o$ and $o_1$ be the mean points of $Q$ and $Q_{1}$, respectively. Then  $||o_{1}-o||\leq\sqrt{\frac{1-\alpha}{\alpha}}\delta$, where $\delta^2=\frac{1}{|Q|}\sum_{q\in Q}||q-o||^2$.
\end{lemma}
\begin{proof}
The following claim has been proved in~\cite{KSS}. \\

\noindent\textbf{Claim 1}
{\em Let $Q$ be a set of points in $\mathbb{R}^d$ space, and $o$ be the mean point of $Q$. For any point $\tilde{o}\in \mathbb{R}^d$, $\sum_{q\in Q}||q-\tilde{o}||^2=\sum_{q\in Q}||q-o||^2+|Q| \times ||o-\tilde{o}||^2$.}\\



Let $Q_2=Q\setminus Q_1$, and $o_{2}$ be its mean point. By  Claim 1, we have the following two equalities.
  \begin{eqnarray} 
    \sum_{q\in Q_{1}}||q-o||^2 &= &\sum_{q\in Q_{1}}||q-o_1||^2+|Q_1| \times ||o_1-o||^2, \label{for-1}\\
%
  \sum_{q\in Q_{2}}||q-o||^2 &=& \sum_{q\in Q_{2}}||q-o_2||^2+|Q_2| \times ||o_2-o||^2. \label{for-2}
 \end{eqnarray} 
 
Let $L=||o_{1}-o_{2}||$. By the definition of mean point, we have 
\begin{eqnarray}
o=\frac{1}{|Q|}\sum_{q\in Q} q=\frac{1}{|Q|}(\sum_{q\in Q_1} q+\sum_{q\in Q_2} q)=\frac{1}{|Q|}(|Q_1|o_1+|Q_2|o_2). 
\end{eqnarray}
Thus the three points $\{o, o_1, o_2\}$ are collinear, while $||o_{1}-o||=(1-\alpha) L$ and $||o_{2}-o||=\alpha L$. Meanwhile, by the definition of $\delta$, we have 
\begin{eqnarray}
\delta^2=\frac{1}{|Q|}(\sum_{q\in Q_{1}}||q-o||^2+\sum_{q\in Q_{2}}||q-o||^2). 
\end{eqnarray}

Combining (\ref{for-1}) and (\ref{for-2}), we have 

\begin{eqnarray}
\delta^2&=& \frac{1}{|Q|}(\sum_{q\in Q_{1}}||q-o_1||^2+ |Q_1| \times ||o_1-o||^2 \nonumber\\
&&+\sum_{q\in Q_{2}}||q-o_2||^2+|Q_2| \times ||o_2-o||^2)\nonumber\\
&\geq& \frac{1}{|Q|}( |Q_1| \times ||o_1-o||^2+|Q_2| \times ||o_2-o||^2)\nonumber\\
&=&\alpha((1-\alpha)L)^2+(1-\alpha)(\alpha L)^2\nonumber\\
 \hspace{0.4in} &=&\alpha(1-\alpha)L^2.
 \end{eqnarray} 
 Thus, we have $L\leq\frac{\delta}{\sqrt{\alpha(1-\alpha)}}$, which means that $||o_{1}-o||=(1-\alpha)
L\leq\sqrt{\frac{1-\alpha}{\alpha}}\delta$.
\qed
\end{proof}


\begin{proof}[\textbf{of Lemma \ref{lem-simplex}}]
We prove this lemma by induction on $j$. 

\noindent\textbf{Base case:} For $j=1$, since $Q_1=Q$, $o_1=o$. Thus, the simplex $\mathcal{V}$ and the grid are all simply the point $o_1$. Clearly $\tau=o_{1}$ satisfies the inequality. 

\noindent\textbf{Induction step:} Assume that the lemma holds for any $j\leq j_0$ for some $j_{0} \ge 1$ ({\em i.e.}, the induction hypothesis). Now we consider  the case of $j=j_0+1$. First, we assume that $\frac{|Q_l |}{|Q|}\geq \frac{\epsilon}{4j}$ for each $1\leq l\leq j$. Otherwise, we can reduce the problem to the case of a smaller $j$ in the following way. Let $I=\{l| 1\leq l\leq j, \frac{|Q_l |}{|Q|}< \frac{\epsilon}{4j}\}$ be the index set of small subsets. Then, $\frac{\sum_{l\in I}|Q_l |}{|Q|}<\frac{\epsilon}{4}$, and $\frac{\sum_{l\not\in I}|Q_l |}{|Q|}\geq 1-\frac{\epsilon}{4}$. By Lemma~\ref{lem-close}, we know that  $||o'-o||\leq\sqrt{\frac{\epsilon/4}{1-\epsilon/4}}\delta$, where $o'$ is the mean point of $\cup_{l\not\in I}Q_l$.  Let $(\delta')^2$ be the variance of $\cup_{l\not\in I}Q_l$. Then, we have $(\delta')^2\le\frac{|Q|}{|\cup_{l\not\in I}Q_l|}\delta^2\leq \frac{1}{1-\epsilon/4}\delta^2$. Thus, if we replace $Q$ and $\epsilon$ by $\cup_{l\not\in I}Q_l$ and $\frac{\epsilon}{16}$, respectively, and find a point $\tau$ such that $||\tau-o'||^2\leq\frac{\epsilon}{16}(\delta')^2\leq \frac{\epsilon/16}{1-\epsilon/4}\delta^2$, then we have 
\begin{eqnarray}
||\tau-o||^2\leq(||\tau-o'||+||o'-o||)^2\leq \frac{\frac{9}{16}\epsilon}{1-\epsilon/4}\delta^2\leq \epsilon\delta^2 ,
\end{eqnarray}
where the last inequality is due to the fact $\epsilon<1$. This means that we can reduce the problem to a problem with the point set $\cup_{l\not\in I}Q_l$ and a smaller $j$ ({\em i.e.}, $j-|I|$). By the induction hypothesis, we know that the reduced problem can be solved, where the new simplex would be a subset of $\mathcal{V}$ determined by $\{o_l\mid 1\leq l\leq j, l\not\in I\}$, and therefore the induction step holds for this case. Note that in general, we do not know $I$, but we can enumerate all the $2^j$ possible combinations to guess $I$ if $j$ is a fixed number as is the case  in the algorithm in Section \ref{sec-peeling}. Thus, in the following discussion, we can assume that $\frac{|Q_l |}{|Q|}\geq \frac{\epsilon}{4j}$ for each $1\leq l\leq j$.

For each $1\leq l\leq j$, since $\frac{|Q_l |}{|Q|}\geq \frac{\epsilon}{4j}$, by Lemma \ref{lem-close}, we know that $||o_l-o||\leq\sqrt{\frac{1- \frac{\epsilon}{4j}}{ \frac{\epsilon}{4j}}}\delta\leq 2\sqrt{\frac{j}{\epsilon}}\delta$. This, together with triangle inequality, implies that  for any $1\leq l, l'\leq j$, 
\begin{eqnarray}
||o_l-o_{l'}||\leq ||o_l-o||+||o_{l'}-o||\leq 4\sqrt{j/\epsilon}\delta. \label{for-sim1}
\end{eqnarray}
Thus, if we pick any index $l_0$, and draw a ball $\mathcal{B}$ centered at $o_{l_0}$ and with radius $r=\max_{1\leq l\leq j}\{||o_l-o_{l_0}||\}\leq 4\sqrt{j/\epsilon}\delta$ (by (\ref{for-sim1})), the whole simplex $\mathcal{V}$ will be inside $\mathcal{B}$. Note that $o=\sum^j_{l=1}\frac{|Q_l|}{|Q|}o_l$, so $o$  lies inside the simplex $\mathcal{V}$.
%
%
%
To guarantee that $o$ is contained by the ball $\mathcal{B}$, we can construct $\mathcal{B}$ only in the ($j-1$)-dimensional space spanned by $\{o_1, \cdots, o_j\}$, rather than the whole $\mathbb{R}^d$ space.  Also, if we build a grid inside $\mathcal{B}$ with grid length $\frac{\epsilon r}{4j}$, {\em i.e.,} generating a uniform mesh with each cell being a $(j-1)$-dimensional hypercube of edge length $\frac{\epsilon r}{4j}$,  the total number of  grid points is no more than $O((\frac{8j}{\epsilon})^j)$. With this grid, we know that for any point $p$ inside $\mathcal{V}$, there exists a grid point $g$ such that $||g-p||\leq \sqrt{j (\frac{\epsilon r}{4j})^2}=\frac{\epsilon}{4\sqrt{j}}r\leq \sqrt{\epsilon}\delta$. This means that we can find a grid point $\tau$ inside $\mathcal{V}$, such that $||\tau-o||\leq\sqrt{\epsilon}\delta$. Thus, the induction step holds, and
the lemma is true for any $j\ge 1$.
\qed
\end{proof}

In the above lemma, we assume that the exact positions of $\{o_1, \cdots, o_j\}$ are known (see Figure~\ref{fig-simplex}). However, in some scenarios ({\em e.g.},  in the Algorithm in Section \ref{sec-peeling}), 
we only know an approximate position of each mean point $o_{i}$ (see Figure \ref{fig-simplex2}). The following lemma shows that an approximate position of  $o$ can still be similarly determined (see Section~\ref{sec-detail_shift} for the proof).\\

\begin{lemma}[Simplex Lemma \Rmnum{2}]
\label{lem-shift}
Let $Q$, $o$, $Q_{l}, o_{l}, 1\le l \le j$, and $\delta$ be defined as in Lemma~\ref{lem-simplex}. 
Let  $\{o'_1, \cdots, o'_j\}$ be  $j$ points in $\mathbb{R}^{d}$ such that $||o'_l-o_l ||\leq L$ for $1\leq l\leq j$ and $L>0$, and $\mathcal{V}'$ be the simplex determined by $\{o'_1, \cdots, o'_j\}$. Then for any $0<\epsilon\leq 1$, it is possible to construct a grid of size $O((8j/\epsilon)^j)$ inside $\mathcal{V}'$ such that at least one grid point $\tau$ satisfies the inequality $||\tau-o||\leq\sqrt{\epsilon}\delta+(1+\epsilon)L$. 
\end{lemma}


\section{Peeling-and-Enclosing Algorithm for $k$-CMeans} 
\label{sec-candidates}

In this section, we present a  new Peeling-and-Enclosing (PnE) algorithm for generating a set of candidates for  the mean points of $k$-CMeans.  Our algorithm uses peeling spheres and the simplex lemma to iteratively find a good candidate for each unknown cluster.  
An overview of the algorithm is given in Section \ref{sec-ov}. 

%



\textbf{Some notations:} Let $P=\{p_1, \cdots, p_n\}$ be the set of $\mathbb{R}^{d}$ points in $k$-CMeans, and $\mathcal{OPT}=\{Opt_1, \cdots,$ $Opt_k\}$ be the $k$ unknown optimal constrained clusters  with  $m_{j}$ being the mean point of  $Opt_j$ for $1\leq j\leq k$. 
Without loss of generality, we assume that $|Opt_1|\geq |Opt_2|\geq\cdots\geq |Opt_k |$. Denote by $\delta^2_{opt}$ the optimal objective value, {\em i.e.,} $\delta^2_{opt}=\frac{1}{n}\sum^k_{j=1}\sum_{p \in Opt_j}||p-m_j||^2$. We also set $\epsilon>0$ as the parameter related to the quality of the approximate clustering result.


\subsection{Overview of the Peeling-and-Enclosing Algorithm} 
\label{sec-ov}

\begin{figure}[h]
\begin{center}
  \hspace{0.15in}\subfloat[]{\label{nucleus}\includegraphics[width=.3\textwidth]{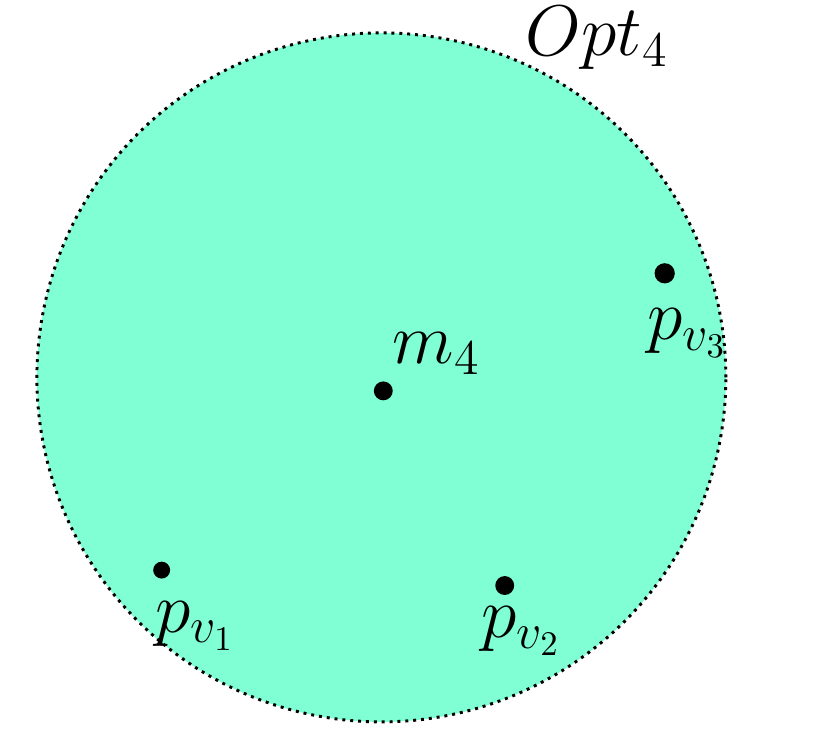}}
  \hspace{0.6in}\subfloat[]{\label{nucleus}\includegraphics[width=.3\textwidth]{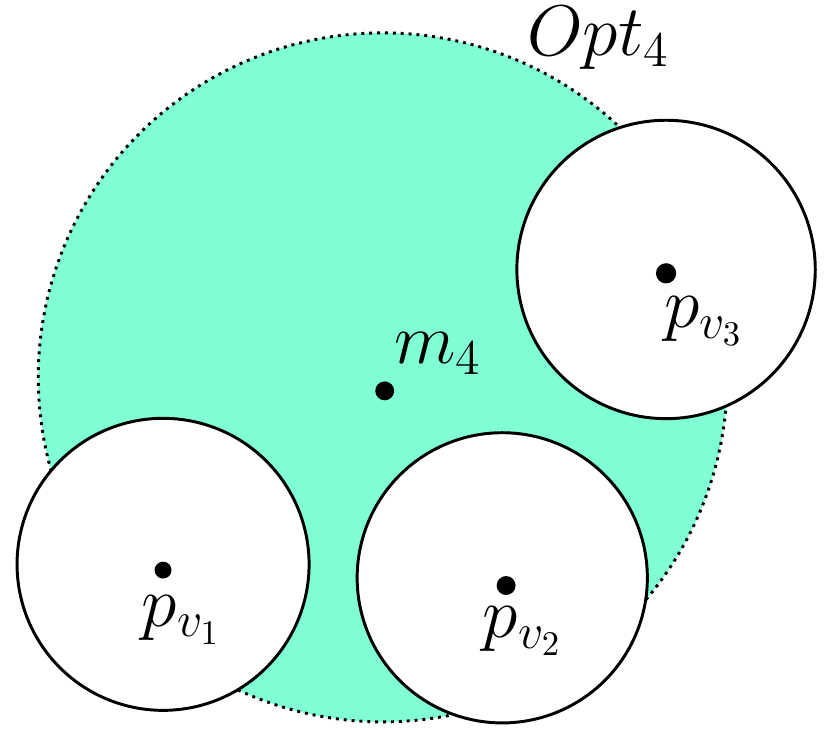}}
  
  \hspace{0.15in}\subfloat[]{\label{nucleus}\includegraphics[width=.3\textwidth]{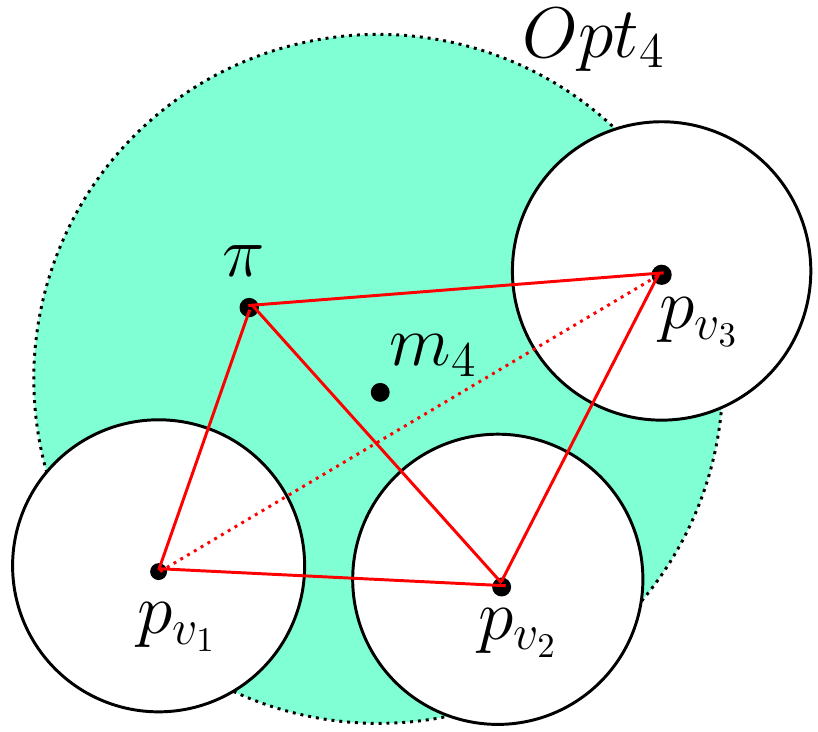}}
  \hspace{0.6in}\subfloat[]{\label{nucleus}\includegraphics[width=.3\textwidth]{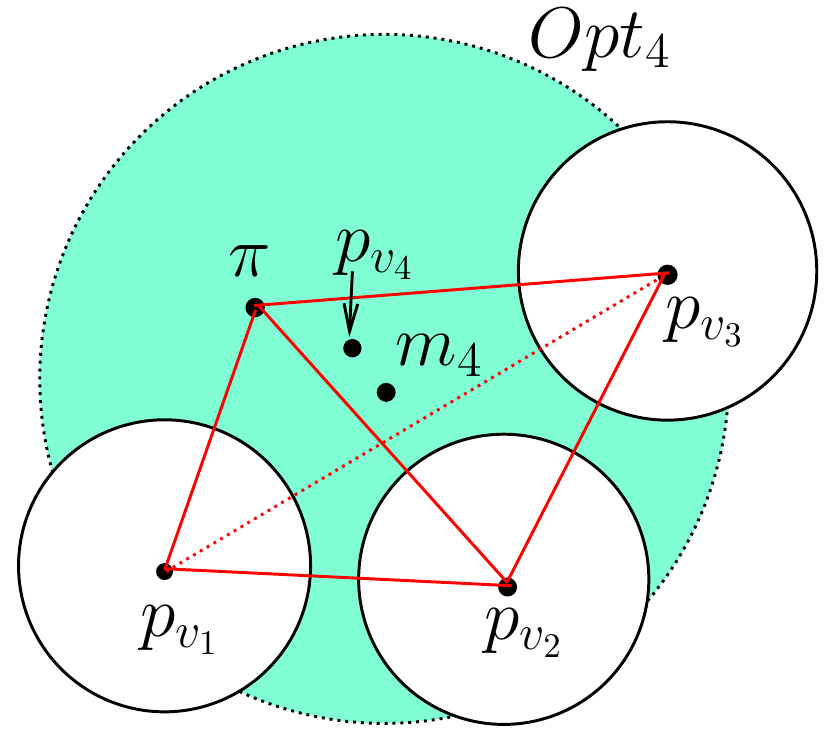}}
  \end{center}
     \caption{Illustration for one iteration of Peeling-and-Enclosing. (a) Beginning of iteration 4; (b) generate 3  spheres (in white) to peel the optimal cluster $Opt_{4}$ (in green); (c) build a simplex (in red) to contain $m_{4}$; (d) find an approximate mean point $p_{v_{4}}$ for $m_{4}$.}
   \label{fig-peeling}
 \end{figure}

Our Peeling-and-Enclosing algorithm needs an upper bound 
$\Delta$ on the optimal value $\delta^2_{opt}$. Specifically, $\delta^2_{opt}$ satisfies the condition $\Delta/c\leq\delta^2_{opt}\leq \Delta$ for some constant $c\geq 1$. In Section \ref{sec-upper}, we will present a novel algorithm to determine such an upper bound for general constrained $k$-means clustering problems. Then, it  searches for a $(1+\epsilon)$-approximation $\delta^{2}$ of $\delta^{2}_{opt}$ in the set 
\begin{eqnarray}
H=\{\Delta/c, (1+\epsilon)\Delta/c, (1+\epsilon)^2\Delta/c, \cdots, (1+\epsilon)^{\lceil\log_{1+\epsilon}c\rceil}\Delta/c\geq \Delta\}.
\end{eqnarray} 
Obviously, there exists one  element of $H$ lying inside the interval $[\delta^{2}_{opt}, (1+\epsilon)\delta^{2}_{opt}]$, and the size of $H$ is $O(\frac{1}{\epsilon}\log c)$.





At each searching step, our algorithm performs a sphere-peeling and simplex-enclosing procedure to iteratively generate $k$ approximate mean points for the constrained clusters. Initially, our algorithm uses Lemmas \ref{lem-dis} and \ref{lem-select} to find an approximate mean point $p_{v_1}$ for $Opt_{1}$  (note that  since $Opt_1$ is the largest cluster, $|Opt_1|/n \geq 1/k$ and the sampling lemma applies). At the $(j+1)$-th iteration,  it already has the approximate mean  points $p_{v_1}, \cdots, p_{v_j}$ for $Opt_1, \cdots, Opt_j$, respectively (see Figure \ref{fig-peeling}(a)). Due to the lack of locality, some points of $Opt_{j+1}$ could be scattered over the regions ({\em e.g.,} Voronoi cells or peeling spheres)  of $Opt_{1}, \cdots, Opt_{j}$ and are difficult to be distinguished from the points in these clusters. 
Since the number of such points could be small (comparing to that of the first $j$ clusters), they need to be handled differently from the remaining points. 
Our idea is to separate them using $j$ peeling spheres, $B_{j+1,1}, \cdots, B_{j+1,j}$, centered at the $j$ approximate mean points respectively and with some properly guessed radius (see Figure \ref{fig-peeling}(b)).
Let $\mathcal{A}$ be the set of unknown points in $Opt_{j+1}\setminus (\cup^j_{l=1}B_{j+1,l})$. Our algorithm considers two cases, (a) $|\mathcal{A}|$ is large enough and (b) $|\mathcal{A}|$ is small. For case (a), since $|\mathcal{A}|$ is large enough, we can use Lemma~\ref{lem-dis} and Lemma~\ref{lem-select} to find an approximate mean point $\pi$ of $\mathcal{A}$, and then construct a simplex determined by $\pi$ and $p_{v_1}, \cdots, p_{v_j}$ to contain the $j+1$-th mean point (see Figure~\ref{fig-peeling}(c)). Note that  $\mathcal{A}$ and  $Opt_{j+1}\cap B_{j+1,l}, 1\le l \le j,$  can be viewed as a partition of $Opt_{j+1}$ where the points covered by multiple peeling spheres can be assigned to anyone of them, and $p_{v_{l}}$ can be shown as an approximate mean point of $Opt_{j+1}\cap B_{j+1,l}$; thus the simplex lemma applies.  For case (b), it directly constructs a simplex determined just by $p_{v_1}, \cdots, p_{v_j}$. For either  case, our algorithm builds a grid inside the simplex and uses Lemma~\ref{lem-shift} to find an approximate mean point for $Opt_{j+1}$ ({\em i.e.}, $p_{v_{j+1}}$, see Figure~\ref{fig-peeling}(d)).  The algorithm repeats the  Peeling-and-Enclosing procedure $k$ times to generate the $k$ approximate mean points.


\subsection{Peeling-and-Enclosing Algorithm}
\label{sec-peeling}
 
Before presenting our algorithm, we first introduce  two basic lemmas from~\cite{IKI,DX14} for random sampling.
Let $S$ be a set of $n$ points in $\mathbb{R}^d$ space, and $T$ be a randomly selected subset of size $t$ from $S$. Denote by $m(S)$ and $m(T)$ the mean points of $S$ and $T$ respectively.

\begin{lemma}[\cite{IKI}]
\label{lem-dis}
With probability $1-\eta$, $||m(S)-m(T)||^2<\frac{1}{\eta t}\delta^2$, where $\delta^2=\frac{1}{n}\sum_{s\in S}||s-m(S)||^2$ and $0<\eta<1$. 
\end{lemma}
 
 \begin{lemma}[\cite{DX14}]
\label{lem-select}
Let $\Omega$ be a set of elements, and $S$ be a subset of $\Omega$ with 
$\frac{|S|}{|\Omega|}=\alpha$ for some $\alpha \in (0,1)$. If we randomly select $\frac{t\ln\frac{t}{\eta}}{\ln(1+\alpha)}=O(\frac{t}{\alpha}\ln\frac{t}{\eta})$ elements from $\Omega$, then with probability at least $1-\eta$, the sample contains at least $t$ elements from $S$ for $0<\eta<1$ and $t\in \mathbb{Z}^+$.
\end{lemma}

Our Peeling-and-Enclosing algorithm is shown in Algorithm~\ref{alg-kcmeans}. 

\begin{algorithm}
   \caption{Peeling-and-Enclosing for $k$-CMeans}
   \label{alg-kcmeans}
\begin{algorithmic}
   \STATE {\bfseries Input:} $P=\{p_1, \cdots, p_n\}$ in $\mathbb{R}^d$, $k\geq 2$, a constant $\epsilon\in(0, \frac{1}{4k^2})$, and an upper bound $\Delta\in [\delta^2_{opt}, c\delta^2_{opt}]$ with $c\geq 1$.
   \STATE {\bfseries Output:} A set of $k$-tuple candidates for the $k$ constrained mean points.
    \begin{enumerate}

\item For $i=0$ to $\lceil\log_{1+\epsilon}c\rceil$ do
\begin{enumerate}
\item Set $\delta=\sqrt{(1+\epsilon)^i\Delta/c}$, and run Algorithm~\ref{alg-tree}.

\item  Let $\mathcal{T}_i$ be the output tree.

\end{enumerate}

\item For each root-to-leaf path of  every $\mathcal{T}_i$, build a $k$-tuple candidate using the $k$ points associated with the path.

\end{enumerate}

\end{algorithmic}
\end{algorithm}

\begin{algorithm}
   \caption{Peeling-and-Enclosing-Tree}
   \label{alg-tree}
\begin{algorithmic}
   \STATE {\bfseries Input:} $P=\{p_1, \cdots, p_n\}$ in $\mathbb{R}^d$, $k\geq 2$, a constant $\epsilon\in(0, \frac{1}{4k^2})$, and $\delta>0$.
\begin{enumerate}
\item Initialize $\mathcal{T}$ as a single root node $v$ associated with no point.
\item Recursively grow each node $v$ in the following way
\begin{enumerate}
\item If the height of $v$ is already $k$, then it is a leaf.
\item Otherwise, let $j$ be the height of $v$. Build the radius candidate set $\mathcal{R}=$ \\ $\cup^{\log n}_{t=0}\{\frac{1+l\frac{\epsilon}{2}}{2(1+\epsilon)}j2^{t/2}\sqrt{\epsilon}\delta\mid 0\le l\le 4+\frac{2}{\epsilon}\}$. For each $r\in\mathcal{R}$, do
\begin{enumerate}
\item   Let $\{p_{v_1}, \cdots, p_{v_j}\}$ be the $j$ points associated with the nodes on the root-to-$v$ path. 

\item For each $p_{v_l}$, $1\leq l\leq j$, construct a ball $B_{j+1,l}$ centered at $p_{v_l}$ and with radius $r$. 
\item Take a random sample from $P\setminus\cup^j_{l=1}B_{j+1,l}$ of size $s=\frac{8k^3}{\epsilon^9}\ln\frac{k^2}{\epsilon^6}$. Compute the mean points of all the subsets of the sample, and denote them by $\Pi=\{\pi_1, \cdots, \pi_{2^s-1}\}$.
\item For each $\pi_i \in \Pi$, construct a simplex using $\{p_{v_1}, \cdots, p_{v_j},$ $ \pi_i\}$ as its vertices. Also construct another simplex using $\{p_{v_1}, \cdots, p_{v_j}\}$ as its vertices. For each simplex, build a grid with
size $O(($ $32j/\epsilon^2)^j)$ inside itself and each of its $2^{j}$ possible degenerated sub-simplices. 

\item In total, there are $2^{s+j} (32j/\epsilon^2)^j$ grid points inside the $2^s$ simplices. For each grid point, add one child to $v$, and associate it with the grid point.
\end{enumerate}

\end{enumerate}
\item Output $\mathcal{T}$.
\end{enumerate}

\end{algorithmic}
\end{algorithm}

 
 \begin{theorem} 
\label{the-ptas}
Let $P$ be the set of $n$ $\mathbb{R}^{d}$ points and $k\in \mathbb{Z}^+$ be a fixed constant. In $O(2^{poly(\frac{k}{\epsilon})}n(\log n)^{k+1} d )$ time,  Algorithm~\ref{alg-kcmeans} outputs  $O(2^{poly(\frac{k}{\epsilon})}(\log n)^{k})$ $k$-tuple candidate mean points. With constant probability,  there exists one $k$-tuple candidate in the output which is able to induce a $\big(1+O(\epsilon)\big)$-approximation of $k$-CMeans (together with the solution for the corresponding Partition step).
\end{theorem}

\begin{remark}
\label{rem-ptas}
(1) To increase the success probability to be close to $1$, {\em e.g.,} $1-\frac{1}{n}$, one only needs to repeatedly run the algorithm $O(\log n)$ times; both the time complexity and the number of $k$-tuple candidates increase by a factor of $O(\log n)$. (2) In general,  the Partition step may be challenging to solve. As shown in Section \ref{sec-application},  the constrained clustering problems considered in this paper admit efficient selection algorithms for their Partition steps. 
\end{remark}


 \subsection{Proof of Theorem \ref{the-ptas}}
\label{sec-proof}

Let  $\beta_j= |Opt_j|/n$,  
and 
 $\delta^2_j=\frac{1}{|Opt_j |}\sum_{p\in Opt_j}||p-m_j||^2$, where $m_j$ is the mean point of $Opt_j$. 
By our assumption in the beginning of Section~\ref{sec-candidates}, we know that $\beta_1\geq\cdots\geq\beta_k$. Clearly, $\sum^k_{j=1}\beta_j =1$ and the optimal objective value $\delta^2_{opt}=\sum^k_{j=1}\beta_j\delta^2_j$.
 \vspace{0.05in}
 
\noindent\textbf{Proof Synopsis:} 
Instead of directly proving Theorem \ref{the-ptas}, we consider the following Lemma~\ref{lem-induction} and Lemma~\ref{lem-equal} which jointly ensure the correctness of Theorem~\ref{the-ptas}. In Lemma \ref{lem-induction}, we show that there exists such a root-to-leaf path in one of the returned trees that its associated $k$ points along the path, denoted by $\{p_{v_1}, \cdots, p_{v_k}\}$, are close enough to the mean points $m_{i}, \cdots, m_{k}$ of the $k$ optimal clusters, respectively.  
 The proof is based on mathematical induction;  each step needs to build a simplex, and applies Simplex Lemma~\Rmnum{2} to bound the error, {\em i.e.,} $||p_{v_j}-m_j ||$ in (\ref{for-induction}).  The error  is estimated  by considering both the local ({\em i.e.,}  the variance of  cluster $Opt_j$) and global ({\em i.e.,} the optimal value $\delta_{opt}$) measurements.  This is a more accurate estimation, comparing to  the widely used Lemma \ref{lem-dis} which considers only the local measurement.  Such an improvement is due to  the increased flexibility in 
the Simplex Lemma \Rmnum{2}, and is a  key to our proof. In Lemma~\ref{lem-equal}, we further show that the $k$ points, $\{p_{v_1}, \cdots, p_{v_k}\}$,  in Lemma \ref{lem-induction} induce a $(1+O(\epsilon))$-approximation of $k$-CMeans.

 \begin{lemma}
\label{lem-induction}
Among all the trees generated by Algorithm~\ref{alg-kcmeans}, with constant probability, there exists at least one tree, $\mathcal{T}_i$, which has a root-to-leaf path with each of its nodes $v_{j}$ at level $j$ ($1\leq j\leq k$) associating with a point $p_{v_j}$ and  satisfying the inequality
\begin{eqnarray}
||p_{v_j}-m_j ||\leq \epsilon\delta_j+(1+\epsilon)j\sqrt{\frac{\epsilon}{\beta_j}}\delta_{opt} . \label{for-induction}
\end{eqnarray}
\end{lemma}

Before proving this lemma, we first show its implication.

\begin{lemma}
\label{lem-equal}
If Lemma \ref{lem-induction} is true, then $\{p_{v_1}, \cdots, p_{v_k}\}$ is able to induce a $(1+O(\epsilon))$-approximation of $k$-CMeans (together with the solution for the corresponding Partition step).
\end{lemma}
\begin{proof}
We assume that Lemma \ref{lem-induction} is true. Then for each $1\leq j\leq k$, we have
\begin{eqnarray}
\sum_{p\in Opt_j}||p-p_{v_j}||^2&=&\sum_{p\in Opt_j}||p-m_j||^2+|Opt_j|\times||m_j-p_{v_j}||^2 \nonumber\\
&\leq& \sum_{p\in Opt_j}||p-m_j||^2+|Opt_j|\times2(\epsilon^2\delta^2_j+(1+\epsilon)^2 j^2\frac{\epsilon}{\beta_j}\delta^2_{opt}) \nonumber\\
&=&(1+2\epsilon^2)|Opt_j |\delta^2_j+2(1+\epsilon)^2 j^2\epsilon n\delta^2_{opt}, \label{for-10}
\end{eqnarray}
where the first equation follows from  Claim 1 in the proof of Lemma \ref{lem-close} (note that $m_j$ is the mean point of $Opt_j$), the inequality follows from Lemma \ref{lem-induction} and the fact that $(a+b)^2\leq 2(a^2+b^2)$ for any two real numbers $a$ and $b$, and the last equality follows from the fact that $\frac{|Opt_j|}{\beta_j}=n$. Summing both sides of (\ref{for-10}) over $j$, we have
 \begin{eqnarray}
 \sum^k_{j=1}\sum_{p\in Opt_j}||p-p_{v_j}||^2 &\leq &\sum^k_{j=1}((1+2\epsilon^2)|Opt_j |\delta^2_j+2(1+\epsilon)^2 j^2\epsilon n\delta^2_{opt})\nonumber\\
&\leq& (1+2\epsilon^2)\sum^k_{j=1}|Opt_j |\delta^2_j+2(1+\epsilon)^2 k^3\epsilon n\delta^2_{opt}\nonumber\\
&=&(1+O(k^3)\epsilon) n\delta^2_{opt}, \label{for-11}
\end{eqnarray}
where the last equation follows from the fact that $\sum^k_{j=1}|Opt_j |\delta^2_j=n\delta^2_{opt}$. By (\ref{for-11}), we know that $\{p_{v_1}, \cdots, p_{v_k}\}$ will induce a $(1+O(k^3)\epsilon)$-approximation for $k$-CMeans (together with the solution for the corresponding Partition step). 
Note that $k$ is assumed to be a fixed number. Thus the lemma is true.
\qed
\end{proof}

Lemma \ref{lem-equal} implies that Lemma \ref{lem-induction} is indeed sufficient to ensure the correctness of Theorem~\ref{the-ptas} (except for the number of candidates and the time complexity). Now we prove Lemma \ref{lem-induction}.




\begin{proof}[\textbf{of Lemma \ref{lem-induction}}]
Let $\mathcal{T}_i$ be the tree generated by Algorithm~\ref{alg-tree} when $\delta$ falls in the interval of $[\delta_{opt}, (1+\epsilon)\delta_{opt}]$.  We will focus our discussion on $\mathcal{T}_{i}$, and prove the lemma by mathematical induction on $j$.

\noindent\textbf{Base case:} For $j=1$, since $\beta_1=\max\{\beta_j |1\leq j\leq k\}$, we have $\beta_1\geq\frac{1}{k}$. By Lemma~\ref{lem-dis} and Lemma~\ref{lem-select}, we can find the approximate mean point through random sampling. Let $\Omega$ and $S$ (in Lemma~\ref{lem-select}) be $P$ and $Opt_1$, respectively. Also, $p_{v_{1}}$ is the mean point of the random sample from $P$. Lemma~\ref{lem-select} ensures that the sample contains enough number of points from $Opt_1$, and Lemma~\ref{lem-dis} implies that $||p_{v_1}-m_1||\leq \epsilon\delta_1\leq \epsilon\delta_1+(1+\epsilon)\sqrt{\frac{\epsilon}{\beta_1}}\delta_{opt}$.

\noindent\textbf{Induction step:} Suppose $j>1$. We assume that there is a path in $\mathcal{T}_i$ from the root to the $(j-1)$-th level, such that for each $1\leq l\leq j-1$, the level-$l$ node $v_{l}$ on the path is associated with a point $p_{v_{l}}$ satisfying the inequality $||p_{v_l}-m_l ||\leq \epsilon\delta_l+(1+\epsilon)l\sqrt{\frac{\epsilon}{\beta_l}}\delta_{opt} $ ({\em i.e.}, the induction hypothesis). Now we consider  the case of  $j$.  Below we will show that there is one child of $v_{j-1}$,  {\em i.e.}, $v_{j}$, such that  its associated point $p_{v_{j}}$ satisfies the inequality $||p_{v_j}-m_j ||\leq \epsilon\delta_j+(1+\epsilon)j\sqrt{\frac{\epsilon}{\beta_j}}\delta_{opt} $. First, we have the following claim (see Section \ref{sec-detail_claim3} for the proof).\\
 

\noindent\textbf{Claim 2}
{\em In the set of radius candidates in  Algorithm~\ref{alg-tree}, there exists one value $r_j\in \mathcal{R}$ such that 
\begin{eqnarray}
r_j\in [j\sqrt{\epsilon/\beta_j}\delta_{opt}, (1+\frac{\epsilon}{2})j\sqrt{\epsilon/\beta_j}\delta_{opt}].
\end{eqnarray}}

  
Now, we construct the $j-1$ peeling spheres, $\{B_{j,1}, \cdots, B_{j,j-1}\}$ as in Algorithm~\ref{alg-tree}. For each $1\leq l\leq j-1$, $B_{j,l}$ is centered at $p_{v_l}$ and with radius $r_j$. By Markov's inequality and the induction hypothesis, we have the following claim (see Section \ref{sec-detail_claim4} for the proof).\\
 
 %
 
\noindent\textbf{Claim 3}
 {\em For each $1\leq l\leq j-1$,  $|Opt_l \setminus (\bigcup^{j-1}_{w=1}B_{j,w})|\le \frac{4\beta_j n}{\epsilon}$.} \\


 Claim 3 shows that $|Opt_l \setminus (\bigcup^{j-1}_{w=1}B_{j,w})|$ is bounded for $1\leq l\leq j-1$, which helps us to find the approximate mean point of $Opt_j $. Induced by the $j-1$ peeling spheres $\{B_{j,1}, \cdots, B_{j,j-1}\}$, $Opt_j$ is divided into $j$ subsets, $Opt_j\cap B_{j,1}$, $\cdots$, $Opt_j\cap B_{j,j-1}$ and $Opt_j \setminus(\bigcup^{j-1}_{w=1}B_{j,w})$. For ease of discussion, let $P_l$ denote $Opt_j\cap B_{j,l}$ for $1\leq l\leq j-1$,  $P_j$ denote $Opt_j \setminus(\bigcup^{j-1}_{w=1}B_{j,w})$, and $\tau_{l}$ denote the mean point of $P_l$ for $1\leq l\leq j$. Note that the peeling spheres may intersect with each other. For any two intersecting spheres $B_{j,l_1}$ and $B_{j,l_2}$, we arbitrarily assign the points in $Opt_j\cap (B_{j,l_1}\cap B_{j,l_2})$  to either $P_{l_1}$ or $P_{l_2}$. Thus, we can assume that $\{P_l\mid 1\leq l\leq j\}$ are pairwise disjoint.

Now consider the size of $P_{j}$. We have the following two cases: (a) $|P_j |\geq \epsilon^3\frac{\beta_{j}}{j}n$ and (b) $|P_j |<\epsilon^3\frac{\beta_{j}}{j}n$. We show how,  in each case, Algorithm~\ref{alg-tree} can obtain an approximate mean point for $Opt_{j}$ by using the simplex lemma ({\em i.e.}, Lemma \ref{lem-shift}).
%
%
  


For case (a),   by  Claim 3, together with the fact that $\beta_l\leq \beta_{j}$ for $l>j$, we know that 
\begin{eqnarray}
\sum^k_{l=1}|Opt_{l}\setminus(\bigcup^{j-1}_{w=1}B_{j,w})|&\leq&\sum^{j-1}_{l=1}|Opt_{l}\setminus(\bigcup^{j-1}_{w=1}B_{j,w})|+|P_j|+\sum^k_{l=j+1}|Opt_l|\nonumber\\
&\leq&\frac{4(j-1)\beta_j}{\epsilon}n+|P_j|+(k-j)\beta_j n, \label{for-pj}
\end{eqnarray}
where the second inequality follows from Claim 3. So we have
\begin{eqnarray}
\frac{|P_j |}{\sum^k_{l=1}|Opt_{l}\setminus(\bigcup^{j-1}_{w=1}B_{j,w})|}&\geq&\frac{|P_j |}{\frac{4(j-1)\beta_j}{\epsilon}n+|P_j |+(k-j)\beta_j n}.
%
\end{eqnarray}
We view the right-hand side as a function of $|P_j|$. 
Given any $h>0$, the function $f(x)=\frac{x}{x+h}$ is an increasing function on the variable $x\in [0, +\infty)$. Note that we assume $|P_j |\geq \epsilon^3\frac{\beta_{j}}{j}n$. Thus 
\begin{eqnarray}
\frac{|P_j |}{\sum^k_{l=1}|Opt_{l}\setminus(\bigcup^{j-1}_{w=1}B_{j,w})|}&\geq&\frac{\frac{\epsilon^3}{j}\beta_j n}{\frac{4(j-1)\beta_j}{\epsilon}n+\frac{\epsilon^3}{j}\beta_j n+(k-j)\beta_j n}\nonumber\\
&>&\frac{\epsilon^4}{8kj}\ge\frac{\epsilon^4}{8k^2}, \label{for-pj2}
\end{eqnarray}
(\ref{for-pj2}) implies that $P_{j}$ is large enough, comparing to the set of points outside the peeling spheres. Hence, we can obtain an approximate mean point $\pi $ for $P_j$ in the following way. First, we set $t=\frac{k}{\epsilon^5}$, $\eta=\frac{\epsilon}{k}$, and take a sample of size $\frac{t\ln(t/\eta)}{\epsilon^4 /8k^2}=\frac{8k^3}{\epsilon^9}\ln\frac{k^2}{\epsilon^6}$.  By Lemma \ref{lem-select}, we know that  with probability $1-\frac{\epsilon}{k}$, the sample contains at least $\frac{k}{\epsilon^5}$ points from $P_j$.  Then we let $\pi$ be the mean point of the $\frac{k}{\epsilon^5}$ points from $P_j$, and $a^{2}$ be the variance of $P_j$. By Lemma \ref{lem-dis}, we know that  with probability $1-\frac{\epsilon}{k}$, $||\pi-\tau_j||^2\leq \epsilon^4 a^2$. Also, since $\frac{|P_j|}{|Opt_j|}=\frac{|P_j|}{\beta_j n}\geq\frac{\epsilon^3}{j}$ (because $|P_j |\geq \epsilon^3\frac{\beta_{j}}{j}n$ for case (a)), we have $a^2\le\frac{|Opt_j|}{|P_j|}\delta^2_j\leq\frac{j}{\epsilon^3}\delta^2_j$. Thus, 
\begin{eqnarray}
||\pi-\tau_j||^2\leq \epsilon j\delta^2_j. \label{for-pitau}
\end{eqnarray}


\begin{figure}[h]
 \centering
  \subfloat[]{\label{fig-case1}\includegraphics[height=1.4in]{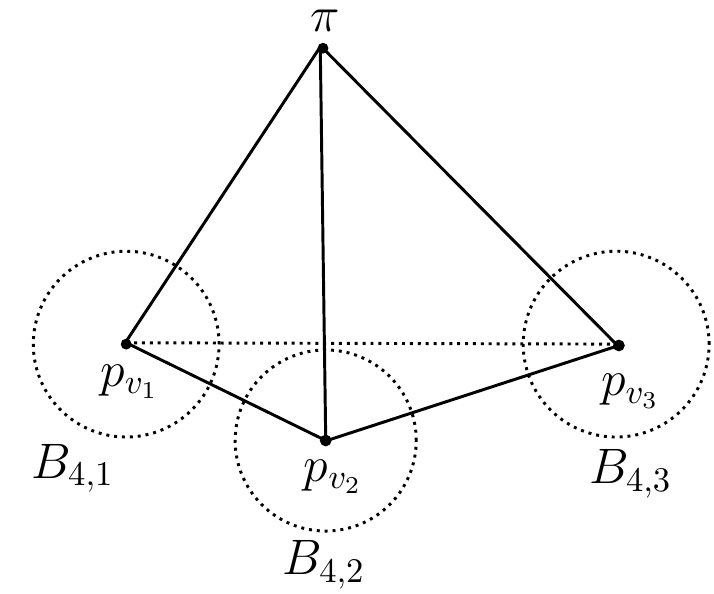}}
  \hspace{0.8in}
  \subfloat[]{\label{fig-case2}\includegraphics[height=1.1in]{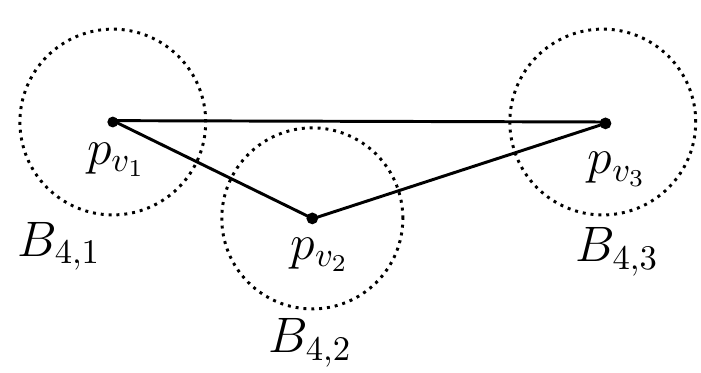}}
   \caption{Figure \ref{fig-case1} and \ref{fig-case2} are the Simplexes of  case (a) and case (b) with $j=4$ respectively.}
\end{figure}


Once obtaining $\pi$,  we can apply Lemma \ref{lem-shift} to find a point $p_{v_{j}}$ satisfying the condition of $||p_{v_{j}}-m_j||\leq \epsilon\delta_j+(1+\epsilon)j\sqrt{\frac{\epsilon}{\beta_j}}\delta_{opt}$. We construct a simplex $\mathcal{V}'_{(a)}$ with vertices $\{p_{v_1}, \cdots, p_{v_{j-1}}\}$ and $\pi$ (see Figure \ref{fig-case1}).  Note that $Opt_j$ is partitioned by  the peeling spheres into $j$ disjoint subsets, $P_1, \cdots, P_j$. Each $P_l$ ($1\le l\le j-1$) lies inside $B_{j,l}$, which implies that $\tau_l$, {\em i.e.,} the mean point of $P_l$, is also inside $B_{j,l}$. Further,  by  Claim 2, for $1\leq l\leq j-1$, we have
 \begin{eqnarray}
||p_{v_l}-\tau_l||\leq r_j\leq (1+\frac{\epsilon}{2})j\sqrt{\epsilon/\beta_j}\delta_{opt}. \label{for-ca1}
\end{eqnarray}
Recall that $\beta_j\delta^2_j\le\delta^2_{opt}$. Thus, together with (\ref{for-pitau}), we have
\begin{eqnarray}
 ||\pi-\tau_j||\leq \sqrt{\epsilon j}\delta_j\leq \sqrt{\epsilon j/\beta_j}\delta_{opt}. \label{for-ca2}
 \end{eqnarray}
 By (\ref{for-ca1}) and (\ref{for-ca2}), if  setting the value of $L$ (in Lemma \ref{lem-shift})  to be 
 \begin{eqnarray}
 \max\{r_j,||\pi-\tau_j|| \}&\le&\max\{(1+\frac{\epsilon}{2})j\sqrt{\epsilon/\beta_j}\delta_{opt}, \sqrt{\epsilon j/\beta_j}\delta_{opt}\}\nonumber\\
 &=&(1+\frac{\epsilon}{2})j\sqrt{\epsilon/\beta_j}\delta_{opt},
 \end{eqnarray}
and the value of $\epsilon$ (in Lemma \ref{lem-shift}) to be $\epsilon_0=\epsilon^2/4$, by Lemma \ref{lem-shift} we can construct a grid inside the simplex $\mathcal{V}'_{(a)}$ with size $O((8j/\epsilon_0)^j)$ to ensure the existence of the grid point $\tau$ satisfying the inequality of 
\begin{eqnarray}
||\tau-m_j||\leq \sqrt{\epsilon_0}\delta_j+(1+\epsilon_0)L\leq\epsilon\delta_j+(1+\epsilon)j\sqrt{\frac{\epsilon}{\beta_j}}\delta_{opt}. 
\end{eqnarray}
Hence, let $p_{v_j}$ be the grid point $\tau$, 
and the induction step  holds for this case.



For case (b), we can also apply Lemma \ref{lem-shift}  to find an approximate mean point in a way similar to case (a); the difference is that we construct a simplex $\mathcal{V}'_{(b)}$ with vertices $\{p_{v_1}, \cdots, p_{v_{j-1}}\}$ (see Figure~\ref{fig-case2}). Roughly speaking, since $|P_j|$ is small, the mean points of $Opt_j\setminus P_j$ and $Opt_j$ are very close to each other (by Lemma \ref{lem-close}). Thus, we can ignore $P_j$ and just consider $Opt_j\setminus P_j$. 
%
%

Let $a^2$ and $m'_j$ denote the variance and mean point of $Opt_j\setminus P_j$ respectively. We know that $\{P_1, P_2, \cdots, P_{j-1}\}$ is a partition on $Opt_j\setminus P_j$. Thus, similar with case (a), we construct a simplex $\mathcal{V}'_{(b)}$ determined by $\{p_{v_1}, \cdots, p_{v_{j-1}}\}$ (see Figure \ref{fig-case2}), set  the value of $L$ to be $r_j\leq (1+\frac{\epsilon}{2})j\sqrt{\frac{\epsilon}{\beta_j}}\delta_{opt}$, and then build a grid inside $\mathcal{V}'_{(b)}$ with size $O((\frac{8j}{\epsilon_0})^j)$, where  $\epsilon_0=\epsilon^2/4$. By Lemma \ref{lem-shift}, we know that there exists one grid point $\tau$ satisfying the condition of 
\begin{eqnarray}
||\tau-m'_j||\leq \sqrt{\epsilon_0}a+(1+\epsilon_0)L\leq\frac{\epsilon}{2}a+(1+\epsilon)j\sqrt{\frac{\epsilon}{\beta_j}}\delta_{opt}. \label{for-cb1}
\end{eqnarray}
Meanwhile, we know that $|Opt_j\setminus P_j|\geq (1-\epsilon^3/j)|Opt_j|$, since $|P_j|\leq \frac{\epsilon^3}{j} |Opt_j|$. Thus, we have $a^2\leq\frac{|Opt_j|}{|Opt_j\setminus P_j|}\delta^2_j \leq \frac{1}{1-\epsilon^3/j} \delta^2_j$, and $||m'_j-m_j||\leq\sqrt{\frac{\epsilon^3/j}{1-\epsilon^3/j}}\delta_j$ (by Lemma \ref{lem-close}). Together with (\ref{for-cb1}), we have
\begin{eqnarray}
||\tau-m_j||&\leq& ||\tau-m'_j||+||m'_j-m_j||\nonumber\\
&\leq&\frac{\epsilon}{2}a+(1+\epsilon)j\sqrt{\frac{\epsilon}{\beta_j}}\delta_{opt}+\sqrt{\frac{\epsilon^3/j}{1-\epsilon^3/j}}\delta_j\nonumber\\
&\leq&\frac{\epsilon}{2} \sqrt{\frac{1}{1-\epsilon^3/j}}\delta_j+(1+\epsilon)j\sqrt{\frac{\epsilon}{\beta_j}}\delta_{opt}+\sqrt{\frac{\epsilon^3/j}{1-\epsilon^3/j}}\delta_j\nonumber\\
&\leq&(\frac{\epsilon}{2} \sqrt{\frac{1}{1-\epsilon^3/j}}+\sqrt{\frac{\epsilon^3/j}{1-\epsilon^3/j}})\delta_j+(1+\epsilon)j\sqrt{\frac{\epsilon}{\beta_j}}\delta_{opt}\nonumber\\
&\leq&\epsilon\delta_j+(1+\epsilon)j\sqrt{\frac{\epsilon}{\beta_j}}\delta_{opt}.\label{for-cb2}
\end{eqnarray}
Hence, let $p_{v_j}$ be the grid point $\tau$, 
and the induction step  holds for this case.



Since Algorithm~\ref{alg-tree} executes every step in our above discussion, the induction step, as well as the lemma, is true.
\qed
\end{proof}

\noindent\textbf{Success probability:} From the above analysis, we know that in the $j$-th iteration, only  case (a) ({\em i.e.}, $|P_j |\geq \epsilon^3\frac{\beta_{j}}{j} n$) needs to consider the success probability of random sampling.
Recall that in case (a), we take a sample of size $\frac{8k^3}{\epsilon^9}\ln\frac{k^2}{\epsilon^6}$. Thus  with probability $1-\frac{\epsilon}{k}$, it contains at least $\frac{k}{\epsilon^5}$ points from $P_j$. Meanwhile, with probability $1-\frac{\epsilon}{k}$, $||\pi-\tau_j||^2\leq \epsilon^4 a^2$. Hence, the success probability in the $j$-th iteration is $(1-\frac{\epsilon}{k})^2$. By taking the union bound, the success probability in all $k$ iterations is $(1-\frac{\epsilon}{k})^{2k}\geq 1-2\epsilon$.

\noindent\textbf{Number of Candidates and Running time:} Algorithm~\ref{alg-kcmeans} calls Algorithm~\ref{alg-tree} $O(\frac{1}{\epsilon}\log c)$ times (in Section~\ref{sec-upper}, we will show that $c$ can be a constant number). It is easy to see that each node in the returned tree has $|\mathcal{R}|2^{s+j} (\frac{32j}{\epsilon^2})^j$ children, where $|\mathcal{R}|=O(\frac{\log n}{\epsilon})$, and $s=\frac{8k^3}{\epsilon^9}\ln\frac{k^2}{\epsilon^6}$. Since the tree has the height of $k$, the complexity of the tree is $O(2^{poly(\frac{k}{\epsilon})}(\log n)^k)$. Consequently, the number of  candidates is $O(2^{poly(\frac{k}{\epsilon})}(\log n)^k)$. Further, since each node takes $O(|\mathcal{R}|2^{s+j} (\frac{32j}{\epsilon^2})^j$ $nd)$ time, the total time complexity of the algorithm is $O(2^{poly(\frac{k}{\epsilon})}$ $n(\log n)^{k+1} d )$.
%
%
%
%
%

\subsection{Upper Bound Estimation}
\label{sec-upper}

As mentioned in Section \ref{sec-ov}, our Peeling-and-Enclosing algorithm needs an upper bound $\Delta$ on the optimal value  $\delta^2_{opt}$. 
To compute this, our main idea is to use some  unconstrained $k$-means clustering algorithm $\mathcal{A}_{*}$  ({\em e.g.,} the linear time $(1+\epsilon)$-approximation algorithm in \cite{KSS}) on the input points $P$ without considering the constraint, to obtain a 
 $\lambda$-approximation to the $k$-means clustering for some constant $\lambda >1$. Let 
$\mathcal{C}=\{c_1, \cdots, c_k\}$ be the set of mean points returned by algorithm $\mathcal{A}_{*}$\footnote{Note that they are different from $\{m_1, \cdots, m_k\}$, which are the mean points of the $k$ optimal constrained clusters $\{Opt_1, \cdots, Opt_k\}$  of $P$ defined in the beginning of Section~\ref{sec-candidates}.}.  
Below, we show that the Cartesian product $[\mathcal{C}]^k=\underbrace{\mathcal{C}\times\cdots\times\mathcal{C}}_{k}$ contains one $k$-tuple, which is an $(18\lambda+16)$-approximation of $k$-CMeans on the same input $P$. Clearly, to select the $k$-tuple from $[\mathcal{C}]^k$ with the smallest objective value, we still need to solve the Partition step on each $k$-tuple to form the desired clusters. Similar to Remark~\ref{rem-ptas}, we refer the reader to Section \ref{sec-application} for the selection algorithms for the considered constrained clustering problems.

%

 \begin{theorem}
\label{the-constant}
Let $P=\{p_1, \cdots, p_n\}$ be the input points of $k$-CMeans, and $\mathcal{C}=\{c_{1},\cdots,c_{k}\}$ be the mean points of a $\lambda$-approximation of the $k$-means clustering on $P$ (without considering the constraint) for some constant $\lambda\geq 1$. Then $[\mathcal{C}]^k$ contains at least one $k$-tuple which is able to induce an $(18\lambda+16)$-approximation of $k$-CMeans (together with the solution for the corresponding Partition step).
\end{theorem}
%

\begin{figure}[h]
 \centering
  \subfloat[]{\label{fig-move}\includegraphics[height=0.75in]{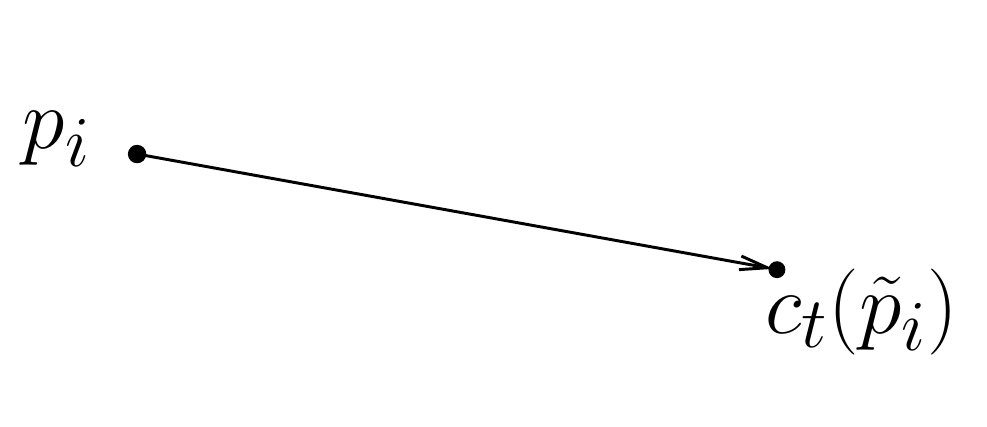}}
  \hspace{0.8in}
  \subfloat[]{\label{fig-tri}\includegraphics[height=1in]{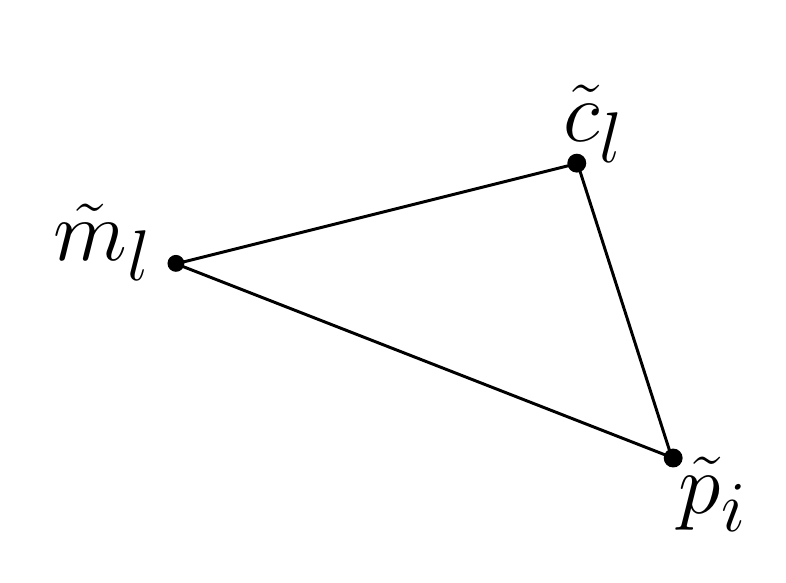}}
   \caption{(a) $p_i$ is moved to $c_t$  and becomes $\tilde{p}_i$; (b) $||\tilde{m}_l-\tilde{c}_l||\leq ||\tilde{m}_l-\tilde{p}_i||$.}
\end{figure}



\noindent\textbf{Proof Synopsis:} Let $\omega$ be the  objective value of the $k$-means clustering on $P$ corresponding to the $k$ mean points in $\mathcal{C}$. 
To prove Theorem~\ref{the-constant}, 
we  create a new instance of $k$-CMeans:  
for each point $p_i\in P$, move it to its nearest point, say $c_{t}$,  in $\{c_1, \cdots, c_k\}$; let $\tilde{p}_i$ denote the new $p_i$ 
 (note that $c_t$ and $\tilde{p}_i$ coincide with each other; see Figure \ref{fig-move}). The set $\tilde{P}=\{\tilde{p}_1, \cdots, \tilde{p}_n\}$ forms a new instance of $k$-CMeans. Let $\tilde{\delta}^2_{opt}$ be the optimal  value of $k$-CMeans on $\tilde{P}$, and $\delta^2_{opt}([\mathcal{C}]^k)$ be the minimum cost of $k$-CMeans on $P$ by restricting its mean points to be one $k$-tuple in $[\mathcal{C}]^k$.
We show that $\tilde{\delta}^2_{opt}$ is bounded by a combination of $\omega$ and $\delta^2_{opt}$, and $\delta^2_{opt}([\mathcal{C}]^k)$ is bounded by a combination of $\omega$ and $\tilde{\delta}^2_{opt}$ (see Lemma~\ref{lem-ieq}). Together with the fact that $\omega$ is no more than $\lambda \delta^2_{opt}$, we consequently obtain that $\delta^2_{opt}([\mathcal{C}]^k)\leq (18\lambda+16)\delta^2_{opt}$, which implies Theorem \ref{the-constant}.

\begin{lemma}
\label{lem-ieq}
$\tilde{\delta}^2_{opt}\leq 2\omega+2\delta^2_{opt}$, and $\delta^2_{opt}([\mathcal{C}]^k)\leq 2\omega+8\tilde{\delta}^2_{opt}$.
\end{lemma}

\begin{proof}
We first prove the  inequality of $\tilde{\delta}^2_{opt}\leq 2\omega+2\delta^2_{opt}$. Consider any point $p_i \in P$. Let $Opt_{l}$ be the optimal cluster containing $p_i$. Then, we have
\begin{eqnarray}
||\tilde{p}_i-m_l||^2 &\leq& (||\tilde{p}_i-p_i||+||p_i-m_l||)^2\nonumber\\
&\leq& 2||\tilde{p}_i-p_i||^2+2||p_i-m_l||^2,\label{for-c1}
\end{eqnarray}
where the first inequality follows from triangle inequality, and the second inequality follows from the fact that $(a+b)^2\leq 2a^2+2b^2$ for any two real numbers $a$ and $b$. For both sides of (\ref{for-c1}), we take the averages over all the points in $P$, and obtain
\begin{eqnarray}
\frac{1}{n}\sum^k_{l=1}\sum_{p_i\in Opt_l}||\tilde{p}_i-m_l||^2 \leq  \frac{2}{n}\sum^n_{i=1}||\tilde{p}_i-p_i||^2+\frac{2}{n}\sum^k_{l=1}\sum_{p_i\in Opt_l}||p_i-m_l||^2.\label{for-c2}
\end{eqnarray}
Note that the left-hand side of (\ref{for-c2}) is not smaller than $\tilde{\delta}^2_{opt}$, since $\tilde{\delta}^2_{opt}$ is the optimal objective value of $k$-CMeans on $\tilde{P}$. For the right-hand side of (\ref{for-c2}), the first term $2\frac{1}{n}\sum^n_{i=1}||\tilde{p}_i-p_i||^2=2\omega$ (by the construction of $\tilde{P}$), and the second term $2\frac{1}{n}\sum^k_{l=1}\sum_{p_i\in Opt_l}||p_i-m_l||^2=2\delta^2_{opt}$. Thus, we have  $\tilde{\delta}^2_{opt}\leq 2\omega+2\delta^2_{opt}$.
 

Next, we show the inequality $\delta^2_{opt}([\mathcal{C}]^k)\leq 2\omega+8\tilde{\delta}^2_{opt}$. Consider $k$-CMeans clustering on $\tilde{P}$.  Let $\{\tilde{m}_1, \cdots, \tilde{m}_k\}$ be the optimal constrained mean points of $\tilde{P}$, and $\{\tilde{O}_1, \cdots, \tilde{O}_k\}$ be the corresponding optimal clusters.
Let $\{\tilde{c}_1, \cdots, \tilde{c}_k\}$ be the $k$-tuple in $[\mathcal{C}]^k$ with $\tilde{c}_l$ being the nearest point in $\mathcal{C}$ to $\tilde{m}_l$. Thus, by an argument similar to the one for (\ref{for-c1}), we have
\begin{eqnarray}
||\tilde{p}_i-\tilde{c}_l||^2 &\leq &2||\tilde{p}_i-\tilde{m}_l||^2+2||\tilde{m}_l-\tilde{c}_l||^2\leq4||\tilde{p}_i-\tilde{m}_l||^2. \label{for-c3}
\end{eqnarray}
for each $\tilde{p}_i\in\tilde{O}_l$. In (\ref{for-c3}), the last one follows from the facts that $\tilde{c}_l$ is the nearest point in $\mathcal{C}$ to $\tilde{m}_l$ and $\tilde{p}_i \in \mathcal{C}$, which implies that 
$||\tilde{m}_l-\tilde{c}_l||\leq ||\tilde{m}_l-\tilde{p}_i||$ (see Figure \ref{fig-tri}).  Summing both sides of (\ref{for-c3}) over all the points in $\tilde{P}$, we have
\begin{eqnarray}
\sum^k_{l=1}\sum_{\tilde{p}_i\in\tilde{O}_l}||\tilde{p}_i-\tilde{c}_l||^2 &\leq&  4\sum^k_{l=1}\sum_{\tilde{p}_i\in\tilde{O}_l}||\tilde{p}_i-\tilde{m}_l||^2. \label{for-c4}
\end{eqnarray}
Now,  consider the following clustering on $P$.  For each $p_i$, if $ \tilde{p}_i\in\tilde{O}_l$, we cluster it 
to the corresponding $\tilde{c}_l$.
 Then the objective value of the clustering is 
\begin{eqnarray}
\frac{1}{n}\sum^k_{l=1}\sum_{\tilde{p}_i\in\tilde{O}_l}||p_i-\tilde{c}_l||^2 &\leq&\frac{1}{n}\sum^k_{l=1}\sum_{\tilde{p}_i\in\tilde{O}_l}(2||p_i-\tilde{p}_i||^2+2||\tilde{p}_i-\tilde{c}_l||^2)\nonumber\\
&\leq& 2\frac{1}{n}\sum^n_{i=1}||p_i-\tilde{p}_i||^2+8\frac{1}{n}\sum^k_{l=1}\sum_{\tilde{p}_i\in\tilde{O}_l}||\tilde{p}_i-\tilde{m}_l||^2.\label{for-c7}
 \end{eqnarray}
The left-hand side of (\ref{for-c7}), $\frac{1}{n}\sum^k_{l=1}\sum_{\tilde{p}_i\in\tilde{O}_l}||p_i-\tilde{c}_l||^2$, is no smaller than $\delta^2_{opt}([\mathcal{C}]^k)$ (by the definition),  and the right-hand side of (\ref{for-c7}) is equal to $2\omega+8\tilde{\delta}^2_{opt}$. Thus, we  have  $\delta^2_{opt}([\mathcal{C}]^k)\leq 2\omega+8\tilde{\delta}^2_{opt}$.
\qed
\end{proof}


\begin{proof}[\textbf{of Theorem \ref{the-constant}}]
By the two inequalities in Lemma \ref{lem-ieq}, we know that $\delta^2_{opt}([\mathcal{C}]^k)\leq 18\omega+16\delta^2_{opt}$. It is obvious that the optimal objective value of the $k$-means clustering is no larger than that of $k$-CMeans on the same set of input points $P$. This implies that $\omega\leq \lambda \delta^2_{opt}$. Thus, we have 
\begin{eqnarray}
\delta^2_{opt}([\mathcal{C}]^k)\leq (18\lambda+16)\delta^2_{opt}.
\end{eqnarray}
So there exists one $k$-tuple in $[\mathcal{C}]^k$, which is able to induce an $(18\lambda+16)$-approximation.
 \qed
\end{proof}

\section{Selection Algorithms for  $k$-CMeans}
\label{sec-application}

As shown in Section \ref{sec-candidates}, a set of $k$-tuple candidates for the mean points of $k$-CMeans can be obtained by our Peeling-and-Enclosing algorithm. To determine the best candidate, we need a selection algorithm to compute the clustering for each $k$-tuple candidate, and select the one with the smallest objective value. Clearly, the key to designing a selection algorithm is to solve the Partition step ({\em i.e.,} generating the clustering) for each $k$-tuple candidate. 
We need to design a problem-specific algorithm for the Partition step, to satisfy the constraint of each individual problem. 

We consider all the constrained $k$-means clustering problems which are mentioned in Section \ref{sec-mainresult}, except for the uncertain data clustering, since Cormode and McGregor~\cite{CM08} have showed that it can be reduced to an ordinary $k$-means clustering problem. However, the $k$-median version of the uncertain data clustering does not have such a property. In Section \ref{sec-application2}, we will discuss how to obtain the $(1+\epsilon)$-approximation by applying our Peeling-and-Enclosing framework.




 
 \begin{figure}[h]
 \centering
  \subfloat[]{\label{fig-gather}\includegraphics[height=1.7in]{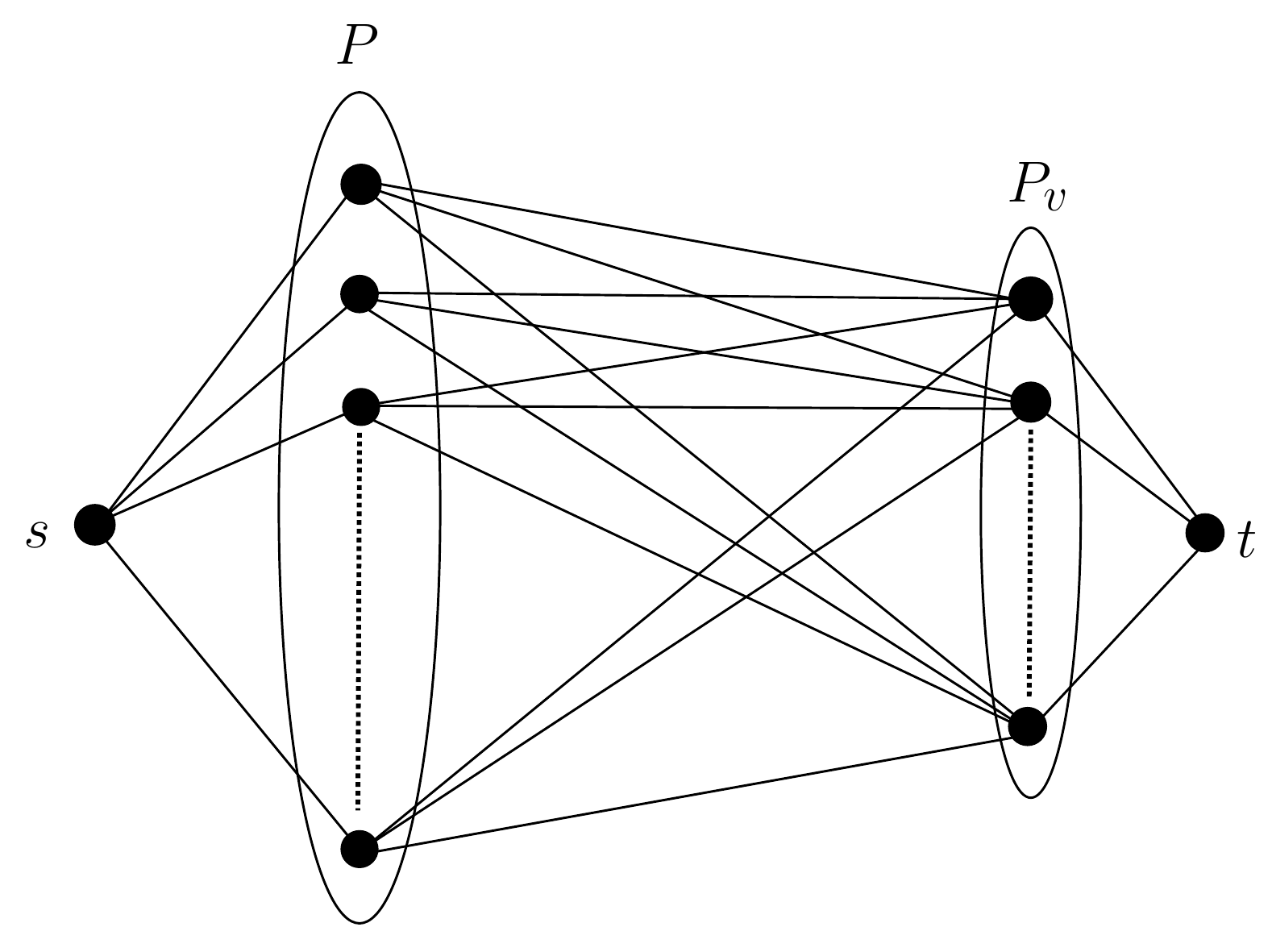}}
  \subfloat[]{\label{fig-diversity}\includegraphics[height=1.7in]{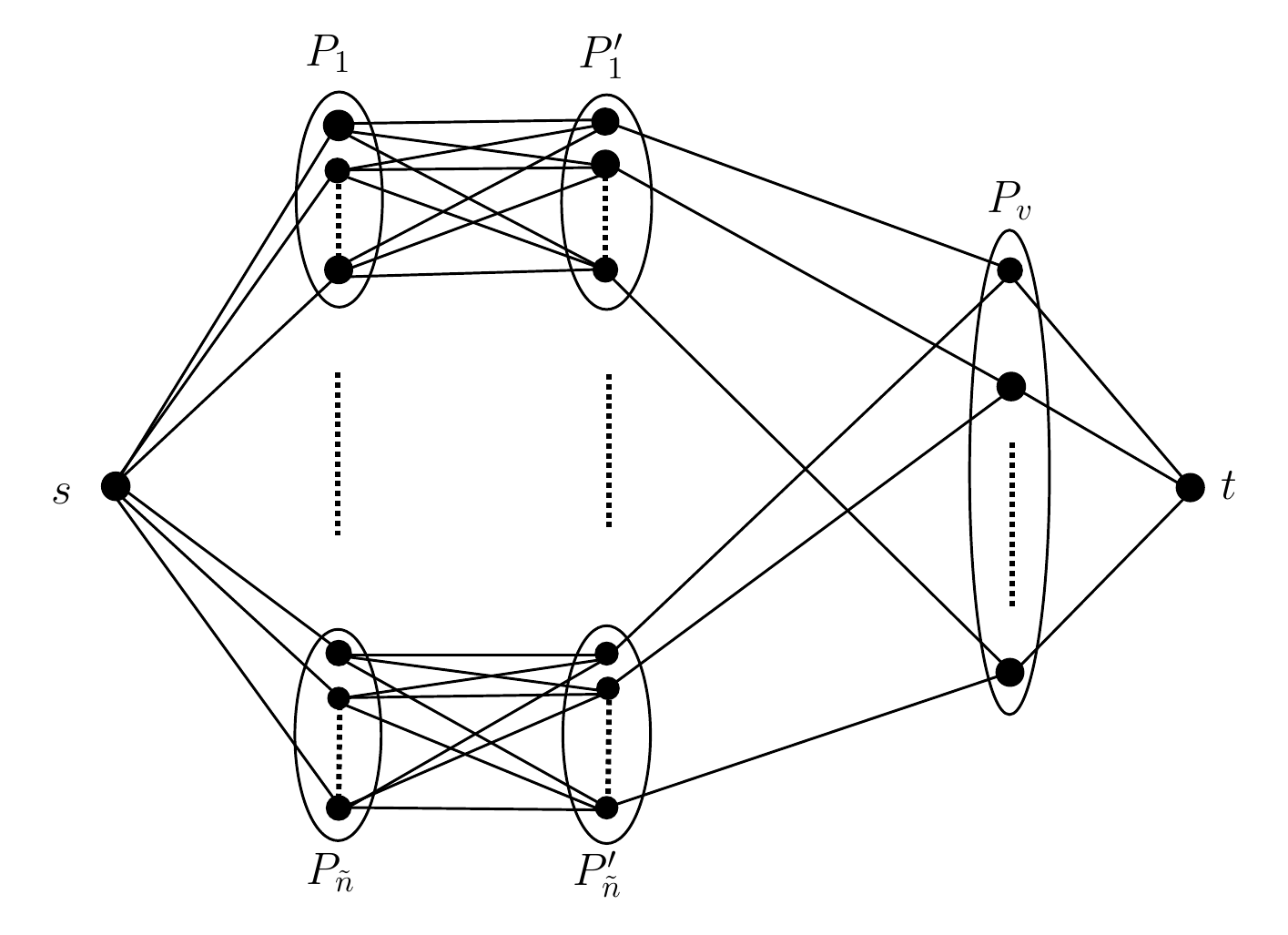}}
   \caption{Minimum cost circulations for $r$-gather clustering (a) and $l$-diversity clustering (b).}
\end{figure}

 \subsection{$r$-Gather $k$-means Clustering}


Let $P$ be a set of $n$ points in $\mathbb{R}^d$.  {\em $r$-Gather $k$-means clustering} (denoted by $(r,k)$-GMeans) on $P$ is the problem of  clustering $P$  into $k$ clusters with size at least $r$, such that the average squared Euclidean distance from each point in $P$ to the mean point of its cluster is minimized \cite{APF}. 


To solve the Partition problem of $(r,k)$-GMeans, we adopt the following strategy.  
For each $k$-tuple candidate $P_{v}=\{p_{v_1}, \cdots p_{v_k}\}$ returned by Algorithm~\ref{alg-kcmeans}, 
build a complete bipartite graph $G$ (see Figure \ref{fig-gather}): each vertex in the left column corresponds to a point in $P$, and each vertex in the right column represents a candidate mean point in $P_{v}$;  for each pair of vertices in different partite sets, connect them by an edge with the weight equal to their squared Euclidean distance. We can solve the Partition problem by finding the minimum cost matching in $G$: each vertex in the left has the supply $1$, and 
each vertex in the right has the demand $r$ and capacity $n$. After adding a source node $s$ connecting to all the verities in the left and a sink node $t$ connecting to all the vertices in the right, we can reduce the Partition problem to a {\em minimum cost circulation} problem, and solve it by using the algorithm in~\cite{Ecourse}. Denote by $V$ and $E$ the sets of vertices and edges of $G$. The running time for solving the minimum cost circulation problem is $O(|E|^2\log |V|+ |E|\cdot|V|\log^2 |V|)$~\cite{Orlin}.  In our case, $|E|=O(n)$ and $|V|=O(n)$ if $k$ is a fixed constant. Also, the time complexity for building $G$ is $O(nd)$. Thus, the total time for solving the Partition problem is $O\Big(n\big(n(\log n)^2+d\big)\Big)$~\footnote{In our problem, an integral solution is necessary for generating the clusters on $P$. Actually, since the demands and capacities are all integers, any optimal solution of the minimum cost circulation problem can be transformed to an integral solution without loss of the quality in $O(|P_v|\cdot |E|)=O(n)$ time~\cite{Ding18}. 
}. Together with the time complexity in Theorem~\ref{the-ptas}, we have the following theorem.

 
\begin{theorem}
\label{the-rgather}
There exists an algorithm yielding a $(1+\epsilon)$-approximation for $(r, k)$-GMeans with constant probability, in $O\Big(2^{poly(\frac{k}{\epsilon})}(\log n)^{k+1}n\big($$n\log n$$+d\big)\Big)$ time. 
\end{theorem}

\subsection{$r$-Capacity $k$-means Clustering}

{\em $r$-Capacity $k$-means clustering} (denoted by $(r, k)$-CaMeans) \cite{KS00} on  a set $P$ of $n$ points in $\mathbb{R}^d$  is the problem of clustering $P$ into $k$ clusters with size at most $r$, such that the average squared Euclidean distance from each point in $P$ to the mean point of its cluster is minimized. 

We can solve the Partition problem  of $(r, k)$-CaMeans in a way similar to that of  $(r, k)$-GMeans; the only difference is that the demand $r$ is replaced by the capacity $r$. 

\begin{theorem}
\label{the-rcapacity}
%
There exists an algorithm yielding a $(1+\epsilon)$-approximation for $(r, k)$-CaMeans with constant probability, in $O\Big(2^{poly(\frac{k}{\epsilon})}(\log n)^{k+1}n\big($$n\log n$$+d\big)\Big)$ time.\end{theorem}

\subsection{$l$-Diversity $k$-means Clustering}

Let $P=\bigcup^{\tilde{n}}_{i=1}P_i$ be a set of colored points in $\mathbb{R}^d$ and $\sum^{\tilde{n}}_{i=1}|P_i|=n$, where the points in each $P_i$ share the same color.
{\em $l$-Diversity $k$-means clustering} (denoted by $(l, k)$-DMeans) on $P$  is the problem of clustering $P$ into $k$ clusters such that the points sharing the same color inside each cluster have a fraction no more than $\frac{1}{l}$ for some $l>1$, and the average squared Euclidean distance from each point in $P$ to the mean point of its cluster is  minimized.

Similar to $(r, k)$-GMeans, we reduce the Partition problem of $(l,k)$-DMeans to a minimum cost circulation problem for each $k$-tuple candidate $P_{v}=\{p_{v_1}, \cdots p_{v_k}\}$. The challenge is that we do not know the size of each resulting cluster, and therefore it is difficult to control the flow on each edge if directly using the bipartite graph built in Figure \ref{fig-gather}. Instead, we add a set of ``gates'' between the input $P$ and the $k$-tuple $P_v$ (see Figure~\ref{fig-diversity}). First, following the definition of $(l, k)$-DMeans, we partition the ``vertices'' $P$ into $\tilde{n}$ groups $\{P_1, \cdots, P_{\tilde{n}}\}$. 
For each $P_i$, we generate a new set of vertices (i.e., the gates) $P'_i=\{c^i_1, \cdots, c^i_k\}$, and connect each pair of $p\in P_i$ and $c^i_j\in P'_i$ by an edge  with  weight  $||p-p_{v_j}||^2$. We also connect each pair of $c^i_j$ and $p_{v_j}$ by an edge with  weight  $0$. In Figure~\ref{fig-diversity}, the size of vertices $|V|=n+k\tilde{n}+k+2=O(kn)$, and the size of edges $|E|=n+kn+k\tilde{n}+k=O(kn)$. Below we show that we can use $c^i_j$ to control the flow from $P_i$ to $p_{v_j}$ by setting appropriate capacities and demands. 

Let $t=\max_{1\leq i\leq \tilde{n}}|P_{i}|$. We consider the value $\lfloor|Opt_j|/l\rfloor$ that is the upper bound on the number of points with the same color in $Opt_j$ (recall $Opt_j$ is the $j$-th optimal cluster defined in Section~\ref{sec-candidates}). The upper bound $\lfloor|Opt_j|/l\rfloor$  can be either between  $1$ and $t$ or larger than $t$. Clearly, if the upper bound is larger than $t$, there is no need to consider the upper bound anymore. Thus, we can enumerate all the $(t+1)^k$ cases to guess the upper bound $\lfloor|Opt_j|/l\rfloor$ for $1\leq j\leq k$. Let $u_{j}$ be the guessed upper bound for $Opt_j$.  If $u_{j}$ is no more than $t$, then each $c^i_j$, $1\leq i\leq \tilde{n}$,  has the capacity $u_j$, and $p_{v_j}$ has the demand $l\times u_j$ and capacity $l\times (u_j+1)-1$. Otherwise ({\em i.e.,} $u_{j} >t$),  set the capacity of each  $c^i_j$,  $1\leq i\leq \tilde{n}$, to be $n$, and the demand and capacity of $p_{v_j}$ to be $l\times (t+1)$ and $n$, respectively. By using the algorithm in \cite{Orlin}, we solve the minimum cost circulation problem for each of the $(t+1)^k$ guesses.

\begin{theorem}
\label{the-diversity}
For any colored point set $P=\bigcup^{\tilde{n}}_{i=1}P_i$ in $\mathbb{R}^{d}$ with $n=|P|$ and $t=\max_{1\leq i\leq \tilde{n}}|P_i|$, there exists an algorithm yielding, 
in $O\Big(2^{poly(\frac{k}{\epsilon})}(\log n)^{k+1}(t+1)^kn\big($$n\log n$$+d\big)\Big)$ time, a $(1+\epsilon)$-approximation for $(l, k)$-DMeans with constant probability.
\end{theorem}

\noindent\textbf{Note:} We can solve  the problem  in \cite{LYZ} by slightly changing the above Partition algorithm. In \cite{LYZ}, it requires that  the size of each cluster is at least $l$ and the points inside each cluster have distinct colors, which means that the upper bound $u_j$ is always equal to $1$ for each $1\leq j\leq k$. Thus,  there is no need to guess the upper bounds in our Partition algorithm. We can simply set the capacity for each $c^i_j$ to be $1$, and the demand for each $p_{v_j}$ to be $l$. With this change, our algorithm yields a $(1+\epsilon)$-approximation with constant probability in $O\Big(2^{poly(\frac{k}{\epsilon})}(\log n)^{k+1}n\big($$n\log n$$+d\big)\Big)$ time.

\subsection{Chromatic $k$-means Clustering}

 Let $P=\bigcup^{\tilde{n}}_{i=1}P_i$ be a set of colored points in $\mathbb{R}^d$ and $\sum^{\tilde{n}}_{i=1}|P_i|=n$, where the points in each $P_i$ share the same color.
 {\em Chromatic $k$-means clustering} (denoted by $k$-ChMeans) \cite{ADH,DX11} on $P$ is the problem of clustering $P$ into $k$ clusters such that no pair of points with the same color is clustered into the same cluster, and the average squared Euclidean distance from each point in $P$ to the mean point of its cluster is  minimized.
  

To satisfy the chromatic requirement, each $P_i$ should have a size no more than $k$. Given a $k$-tuple candidate $P_v=\{p_{v_1}, \cdots, p_{v_k}\}$, we can consider the partition problem for each $P_i$ independently, since there is no mutual constraint among them. It is easy to see that finding a partition of $P_i$ is equivalent to computing a minimum cost one-to-one matching between $P_i$ and $P_v$, where the cost of the matching between any $p\in P_i$ and $p_{v_j}\in P_v$ is their squared Euclidean distance. We can build this bipartite graph in $O(k^2d)$ time, and solve this matching problem by using Hungarian algorithm in $O(k^3)$ time. Thus, the running time of the Partition step for each $P_v$ is $O(k^2(k+d)n)$.

\begin{theorem}
\label{the-chromatic}
%
There exists an algorithm yielding a $(1+\epsilon)$-approximation for $k$-ChMeans with constant probability, in $O\big(2^{poly(\frac{k}{\epsilon})}(\log n)^{k+1}nd\big)$ time. 

\end{theorem}

\subsection{Fault Tolerant $k$-means Clustering}

{\em Fault Tolerant $k$-means clustering} (denoted by $(l,k)$-FMeans) \cite{SS03} on a set $P$ of $n$ points in $\mathbb{R}^d$ and a given integer $1\leq l\leq k$ is the problem of 
finding $k$ points $\mathcal{C}=\{c_1, \cdots, c_k\}\subset\mathbb{R}^d$, such that the average of the total squared distances from each point in $P$ to its $l$ nearest points in $\mathcal{C}$ is minimized. 

To solve the Partition problem of $(l,k)$-FMeans, our idea is to reduce $(l,k)$-FMeans to $k$-ChMeans, and use the Partition algorithm for $k$-ChMeans to generate the desired clusters. The reduction simply makes $l$ monochromatic copies $\{p^1_i, \cdots, p^l_i\}$ for each $p_i\in P$.
%
The following lemma shows the relation of the two problems.

\begin{lemma}
\label{lem-freduce}
For any constant $\lambda \geq1$, a $\lambda$-approximation of $(l, k)$-FMeans on $P$ is equivalent to a $\lambda$-approximation of $k$-ChMeans on $\bigcup^n_{i=1}\{p^1_i, \cdots, p^l_i\}$.
\end{lemma}
\begin{proof}
We build a bijection between the solutions of $(l, k)$-FMeans and $k$-ChMeans. First, we consider the mapping from $(l,k)$-FMeans to $k$-ChMeans. Let $\mathcal{C}=\{c_1, \cdots, c_k\}$ be the $k$ mean points of $(l, k)$-FMeans, and $\{c_{i(1)}, \cdots, c_{i(l)}\}$ $\subset \mathcal{C}$ be the $l$  nearest mean points to each $p_i \in P$.  If using $\mathcal{C}$ as the $k$ mean points of $k$-ChMeans on $\bigcup^n_{i=1}\{p^1_i, \cdots, p^l_i\}$, the $l$ copies $\{p^1_i, \cdots, p^l_i\}$ of $p_{i}$ will be respectively clustered to the $l$ clusters of $\{c_{i(1)}, \cdots, c_{i(l)}\}$ to minimize the cost. 

Now consider the mapping from $k$-ChMeans to $(l,k)$-FMeans. Let $\mathcal{C}'=\{c'_1, \cdots, c'_k\}$ be the $k$ mean points of $k$-ChMeans. For each $i$, $\{c'_{i(1)}, \cdots, c'_{i(l)}\}$ are the mean points of the $l$ clusters that $\{p^1_i, \cdots, p^l_i\}$ are clustered to.  It is easy to see that  the $l$ nearest mean points of $p_{i}$ are $\{c'_{i(1)}, \cdots, c'_{i(l)}\}$ if we use $\mathcal{C}'$ as the $k$ mean points of $(l,k)$-FMeans. 
 
With this bijection, we can pair up the solutions to the two problems. Clearly, each pair of  solutions to $(l, k)$-FMeans and $k$-ChMeans formed by the bijection have the equal objective value. Further, their optimal objective values are equal to each other, and for any pair of solutions, their approximation ratios are the same. Thus, Lemma~\ref{lem-freduce} is true.
\qed
\end{proof}

With Lemma \ref{lem-freduce}, we immediately have the following theorem.

\begin{theorem}
\label{the-ft}

There exists an algorithm yielding a $(1+\epsilon)$-approximation for $(l, k)$-FMeans with constant probability, in $O\big(2^{poly(\frac{k}{\epsilon})}(\log n)^{k+1}nd\big)$ time.

\end{theorem}

\noindent\textbf{Note:} As mentioned in \cite{HHL}, a more general version of fault tolerant clustering problem is to allow each point $p_i \in P$ to have an individual $l$-value $l_{i}$. 
From the above discussion, it is easy to see that this general version can also be solved in the same way ({\em i.e.,} through reduction to  $k$-ChMeans)  and achieve the same approximation result. 

\subsection{Semi-Supervised $k$-means Clustering}

As shown in Section~\ref{sec-mainresult}, semi-supervised clustering has various forms. In this paper, we consider the semi-supervised $k$-means clustering problem which takes into account the geometric cost and priori knowledge. 
Let $P$ be a set of $n$ points in $\mathbb{R}^d$, and $\overline{\mathcal{S}}=\{\overline{S}_1,\cdots, \overline{S}_k \}$ be a given clustering of $P$. 
{\em Semi-supervised $k$-means clustering} (denoted by $k$-SMeans) on $P$ and $\overline{\mathcal{S}}$ is 
the problem of finding a clustering $\mathcal{S}=\{S_1,\cdots,S_k \}$ of $P$ such that the following objective function is minimized,
\begin{eqnarray}
\alpha \frac{Cost(\mathcal{S})}{E_1}+(1-\alpha) \frac{dist\{\mathcal{S}, \overline{\mathcal{S}}\}}{E_2}, \label{for-hybrid}
\end{eqnarray}
where $\alpha \in [0,1]$ is a given constant, $E_1$ and $E_2$ are two given scalars to normalize the two terms, $Cost(\mathcal{S})$ is the $k$-means clustering cost of $\mathcal{S}$, and $dist\{\mathcal{S}, \overline{\mathcal{S}}\}$ is the distance between $\mathcal{S}$ and $\overline{\mathcal{S}}$  defined in the same way as in \cite{BBG}. For any pair of $S_j$ and $\overline{S}_i$, $1\leq j, i\leq k$, their  difference is  $|S_{j} \setminus \overline{S}_i |$.  Given a bipartite matching $\sigma$ between
$\mathcal{S}$ and $\overline{\mathcal{S}}$, $dist\{\mathcal{S}, \overline{\mathcal{S}}\}$ is  defined as $\sum^k_{j=1}|S_{j} \setminus \overline{S}_{\sigma (j)}|$. 

The challenge is that the bipartite matching $\sigma$ is unknown in advance. We fix the $k$-tuple candidate $P_{v}=\{p_{v_1}, \cdots p_{v_k}\}$. To find the desired $\sigma$ to minimize the objective function (\ref{for-hybrid}), we build a bipartite graph, where the left (resp., right) column contains $k$ vertices corresponding to $p_{v_1}, \cdots, p_{v_k}$ (resp., $\overline{S}_1,\cdots, \overline{S}_k $), respectively. For each pair $(p_{v_j}, \overline{S}_i)$, we connect them by an edge; we calculate the edge weight $w_{(i,j)}$ in the following way. For each $p\in \overline{S}_i$, it could be potentially assigned to any of the $k$ clusters in $\mathcal{S}$; if $i=\sigma (j)$, the induced $k$ costs of $p$ will be 
$\{c^1_p, c^2_p, \cdots, c^k_p\}$,
where $c^l_p=\alpha \frac{||p-p_{v_l}||^2}{E_1}$ if $l=j$, or $c^l_p=\alpha \frac{||p-p_{v_l}||^2}{E_1}+(1-\alpha)\frac{1}{E_2}$ otherwise. Thus, we set 
\begin{eqnarray}
w_{(i,j)}=\sum_{p\in \overline{S}_i}\min_{1\leq l\leq k}c^l_p. \label{for-semi_cost}
\end{eqnarray}
We solve the minimum cost bipartite matching problem to determine $\sigma$.
%
%
%
To build such a bipartite graph, we need to first compute all the $kn$ distances from the points in $P$ to the $k$-tuple $P_{v}$; then, we calculate the $k^2$ edge weights via (\ref{for-semi_cost}).
The bipartite graph can be built in a total of $O(k nd+k^2n)$ time, and the optimal matching can be obtained via Hungarian algorithm in $O(k^3)$ time.


\begin{theorem}
\label{the-semi}

There exists an algorithm yielding a $(1+\epsilon)$-approximation for $k$-SMeans with constant probability, in $O\big(2^{poly(\frac{k}{\epsilon})}(\log n)^{k+1}nd\big)$ time. 

\end{theorem}

\section{Constrained $k$-Median Clustering ($k$-CMedian)}
\label{sec-kmedian}

In this section, we extend our approach for $k$-CMeans to the constrained $k$-median clustering problem ($k$-CMedian). Similar to $k$-CMeans, we show that the Peeling-and-Enclosing framework can be used to construct a set of candidates for the constrained median points. Combining this with the selection algorithms (with trivial modification) in Section \ref{sec-application}, we achieve the $(1+\epsilon)$ approximations for  a class of $k$-CMedian problems.

To solve $k$-CMedian, a straightforward idea is to extend the simplex lemma to median points and combine it with the Peeling-and-Enclosing framework to achieve an approximate solution. However, due to the essential difference between mean and median points, such an extension for the simplex lemma  is not always possible. The main reason is that the median point  
({\em i.e.,} Fermat point) does not necessarily lie inside the simplex, and thus there is no guarantee to find the median point by searching inside the simplex. Below is an example showing that the median point  actually can lie outside the simplex.

\begin{figure}[h]
  \centerline{
  \includegraphics[height=1in]{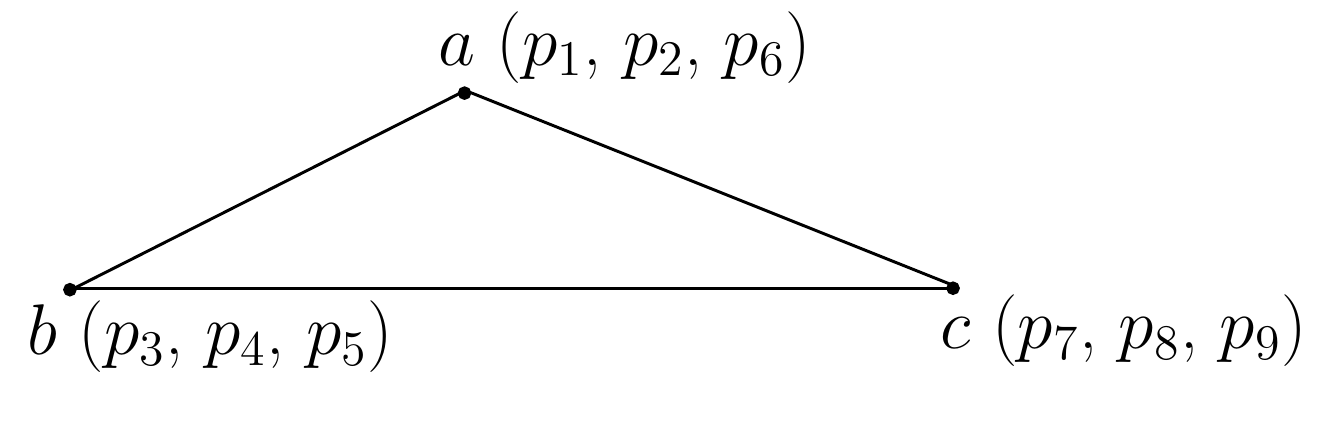}}
    \caption{An example showing non-existence of a simplex lemma for $k$-CMedian. }
  \label{fig-counter}
  \end{figure}

Let $P=\{p_1, p_2, \cdots, p_{9}\}$ be a set of points in $\mathbb{R}^d$. We consider the following partition of $P$,  $P_1=\{p_i\mid 1\leq i\leq 5\}$ and $P_2=\{p_i\mid 6\leq i\leq 9\}$. Assume that all the points of $P$ locate at the three vertices of a triangle $\Delta abc$. Particularly, $\{p_1, p_2, p_6\}$ coincide with vertex $a$, $\{p_3,p_4,p_5\}$   with vertex $b$, and $\{p_7,p_8,p_9\}$  with vertex $c$ (see Figure\ref{fig-counter}). It is easy to see that the median points of $P_1$ and $P_2$ are $b$ and $c$, respectively.  If  the angle $\angle bac\geq \frac{2\pi}{3}$, the median point of $P$ is  vertex $a$ (note that the median point can be viewed as the Fermat point of $\Delta abc$ with each vertex associated with weight $3$). This means that the median point of $P$ is outside the simplex formed by the median points of $P_1$ and $P_2$ ({\em i.e.,} segment $\overline{bc}$). Thus, a good  approximation of the median point cannot be obtained by searching a grid inside $\overline{bc}$. 

To overcome this difficulty, we show that a  weaker version of the simplex lemma exists for median, which enables us to achieve similar results for $k$-CMedian. 

\subsection{Weaker Simplex Lemma for Median Point}
\label{sec-wsl}
Comparing to the simplex lemma in Section \ref{sec-simplex}, the following Lemma \ref{lem-wsl} has two differences. One is that the lemma requires a partial partition on a significantly large subset of $P$,  rather than a complete partition on $P$. Secondly, the grid is built in the flat spanned by $\{o_1, \cdots, o_j\}$, instead of the simplex. Later, we will show that the grid is actually built in a surrounding region of the simplex, and thus the lemma is called ``{\em weaker simplex lemma}''.


\begin{lemma}[Weaker Simplex Lemma]
\label{lem-wsl}
Let $P$ be a set of $n$ points in $\mathbb{R}^d$, and $\bigcup^j_{l=1}P_l\subset P$ be a partial partition of $P$ with $P_{l_1}\cap P_{l_2}=\emptyset$ for any $l_1\neq l_2$. Let $o_{l}$ be the median point of $P_{l}$ for $1\le l \le j$, and $\mathcal{F}$ be the flat spanned by $\{o_1, \cdots, o_j\}$. If $|P\setminus(\bigcup^j_{l=1}P_l)|\leq \epsilon|P|$ for some constant $\epsilon\in (0,1/5)$ and each $P_l$ is contained inside a ball $\mathcal{B}(o_l, L)$ centered at $o_{l}$ and with radius $L\geq 0$, 
then it is possible to build a grid in $\mathcal{F}$ with size $O(j^2(\frac{j\sqrt{j}}{\epsilon})^j)$ such that at least one grid point $\tau$ satisfies the following inequality, where $o$ is the median point of $P$ (see Figure~\ref{fig-weaker}).
\begin{eqnarray}
\frac{1}{|P|}\sum_{p\in P} ||\tau-p||\leq (1+\frac{9}{4}\epsilon)\frac{1}{|P|}\sum_{p\in P}||p-o||+(1+\epsilon)L. \label{for-wsl}
\end{eqnarray}
\end{lemma}


\begin{figure}[]
 \centering
  \includegraphics[height=1.3in]{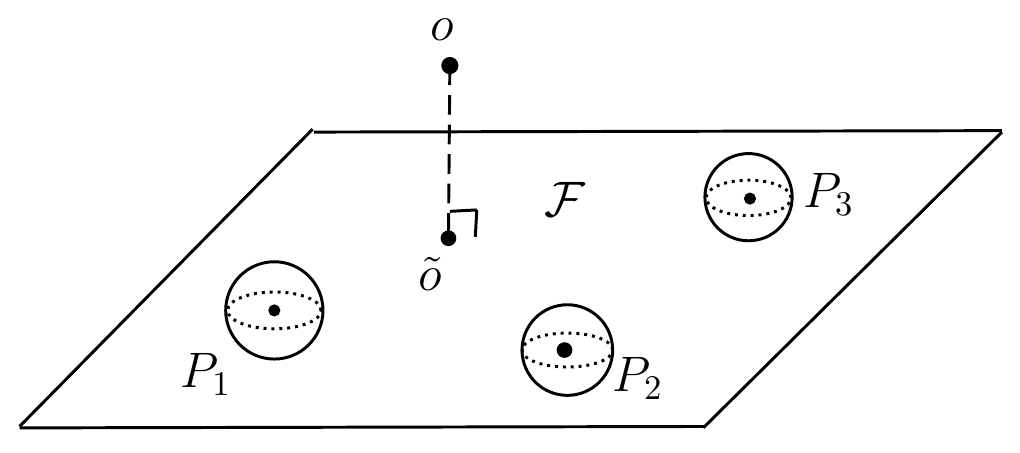}
        \caption{An illustration for Lemma \ref{lem-wsl}.}
  \label{fig-weaker}
 \end{figure}




\noindent\textbf{Proof Synopsis:} To prove Lemma \ref{lem-wsl}, we let $\tilde{o}$ be the orthogonal projection of $o$ to 
 $\mathcal{F}$ (see Figure\ref{fig-weaker}). In  Claim 4, we show that the distance between $o$ and $\tilde{o}$ is bounded, and consequently, the induced cost of $\tilde{o}$, {\em i.e.,} $\frac{1}{|P|}\sum_{p\in P}||p-\tilde{o}||$, is also bounded according to  Claim 5. Thus, $\tilde{o}$ is a good approximation of $o$, and we can focus on building a grid inside $\mathcal{F}$ to approximate $\tilde{o}$. Since  $\mathcal{F}$ is unbounded, we need to determine a range for the grid.  Claim 6 resolves the issue. It considers two cases. One is that there are at least two subsets in the partial partition, $\{P_1, \cdots, P_j\}$, having large enough fractions of $P$; the other is that only one subset is large enough. For either case,  Claim 6 shows that we can determine the range of the grid using the location information of $\{o_1, \cdots, o_j\}$. Finally, we can obtain the desired grid point $\tau$ in the following way: draw a set of balls centered at $\{o_1, \cdots, o_j\}$ with proper radii;  build the grids inside each of the balls, and find the desired grid point $\tau$ in one of these balls.  Note that since all the balls are inside $\mathcal{F}$,  the complexity of the union of the grids is independent of the dimensionality $d$. \\





%

\noindent\textbf{Claim 4}
{\em \begin{eqnarray}
||o-\tilde{o}||\leq  L+\frac{1}{1-\epsilon} \frac{1}{|P|}\sum_{p\in P}||o-p||.\label{for-wsl4}
\end{eqnarray}}\\
%

\begin{proof}
\begin{figure}[]
 \centering
  \includegraphics[height=1in]{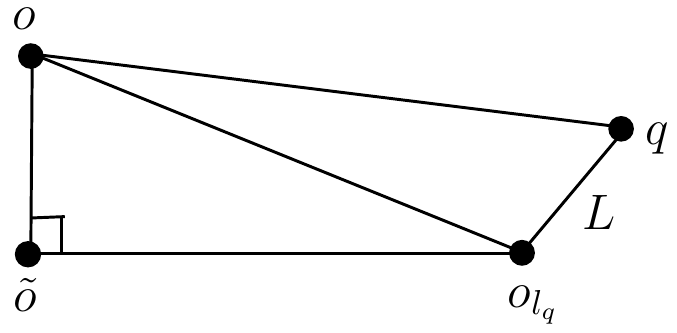}
        \caption{An illustration for  Claim 4.}
  \label{fig-claim5}
 \end{figure}
Lemma~\ref{lem-wsl} assumes that $\bigcup^j_{l=1}P_l\geq (1-\epsilon)|P|$. By Markov's inequality, we know that there exists one point $q\in \bigcup^j_{l=1}P_l$ such that
\begin{eqnarray}
||q-o||\leq \frac{1}{1-\epsilon} \frac{1}{|P|}\sum_{p\in P}||o-p||.\label{for-wsl3}
\end{eqnarray}
Let $P_{l_{q}}$ be the subset containing $q$.  Then from (\ref{for-wsl3}), we  immediately have
\begin{eqnarray}
||o-\tilde{o}||&\leq& ||o_{l_q}-o|| \nonumber\\
&\leq& ||o_{l_q}-q||+||q-o||\nonumber\\
&\leq& L+\frac{1}{1-\epsilon} \frac{1}{|P|}\sum_{p\in P}||o-p||. \label{for-wsl6}
\end{eqnarray}
This implies  Claim 4 (see Figure \ref{fig-claim5}).
\qed
\end{proof}


\noindent\textbf{Claim 5}
{\em \begin{eqnarray}
\frac{1}{|P|}\sum_{p\in P}||p-\tilde{o}|| \leq \frac{1}{1-\epsilon}\frac{1}{|P|}\sum_{p\in P}||p-o||+L. \label{for-wsl7}
\end{eqnarray}}\\

\begin{proof}
\begin{figure}[]
 \centering
  \includegraphics[height=1in]{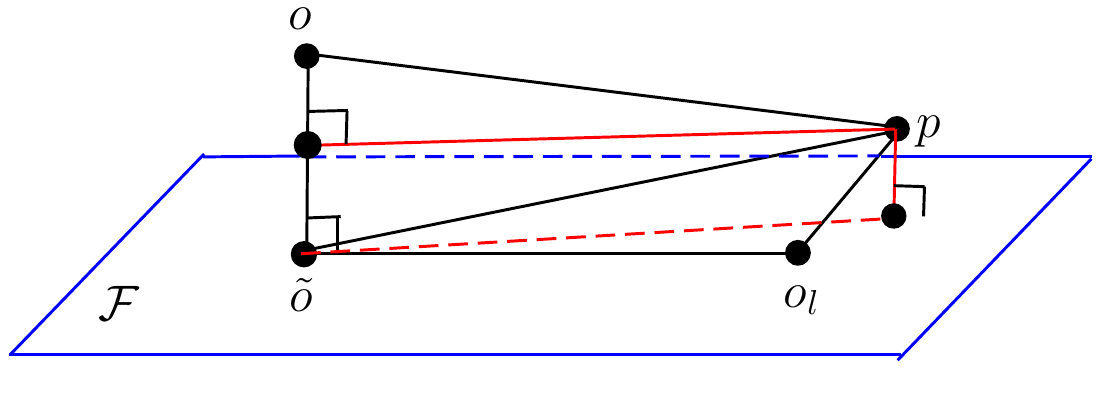}
        \caption{ An illustration for  Claim 5.}
  \label{fig-claim6}
 \end{figure}
For any point $p\in P_l$, let $dist\{\overline{o\tilde{o}}, p\}$ (resp., $dist\{\overline{\mathcal{F}}, p\}$) denote its distance to the line 
$\overline{o\tilde{o}}$ (resp., flat $\mathcal{F}$). See Figure \ref{fig-claim6}.
Then we have
\begin{eqnarray}
||p-\tilde{o}||&=&\sqrt{dist^2\{\overline{o\tilde{o}}, p\}+dist^2\{\mathcal{F}, p\}}, \label{for-w1}\\
||p-o||&\geq& dist\{\overline{o\tilde{o}}, p\}. \label{for-w2}
\end{eqnarray} 
Combining (\ref{for-w1}) and (\ref{for-w2}), we have
\begin{eqnarray}
||p-\tilde{o}||-||p-o||&\leq&\sqrt{dist^2\{\overline{o\tilde{o}}, p\}+dist^2\{\mathcal{F}, p\}}- dist\{\overline{o\tilde{o}}, p\}\nonumber\\
&\leq& dist\{\mathcal{F},p\}\nonumber\\
&\leq& ||p-o_l||\leq L.\label{for-wsl1}
\end{eqnarray}


For any point $p\in P\setminus(\bigcup^j_{l=1}P_l)$, we have
\begin{eqnarray}
||p-\tilde{o}||&\leq& ||p-o||+||o-\tilde{o}||. \label{for-wsl2}
\end{eqnarray}
Combining (\ref{for-wsl1}), (\ref{for-wsl2}), and (\ref{for-wsl4}), we have
\begin{eqnarray}
\frac{1}{|P|}\sum_{p\in P}||p-\tilde{o}||&=&\frac{1}{|P|}(\sum_{p\in \bigcup^j_{l=1}P_l}||p-\tilde{o}||+\sum_{p\in  P\setminus(\bigcup^j_{l=1}P_l)}||p-\tilde{o}||)\nonumber\\
&\leq&\frac{1}{|P|}(\sum_{p\in \bigcup^j_{l=1}P_l}(L+||p-o||)+\sum_{p\in  P\setminus(\bigcup^j_{l=1}P_l)}(||p-o||+||o-\tilde{o}||))\nonumber\\
&\leq&(1-\epsilon)L+\frac{1}{|P|}\sum_{p\in P}||p-o||+\epsilon L+\frac{\epsilon}{1-\epsilon}\frac{1}{|P|}\sum_{p\in P}||p-o||\nonumber\\
&=& \frac{1}{1-\epsilon}\frac{1}{|P|}\sum_{p\in P}||p-o||+L. \label{for-wsl5}
\end{eqnarray}
Thus the claim is true.
\qed
\end{proof}

\noindent\textbf{Claim 6}
{\em At least one of the following  two statements is true. 
\vspace{0.05in}

 \begin{enumerate}
 \item There exist  at least two points in $\{o_1, \cdots, o_j\}$ whose distances to $\tilde{o}$ are no more than $L+\frac{3j}{1-\epsilon}\frac{1}{|P|}\sum_{p\in P}$ $||p-o||$.
 \item There exists one point in $\{o_1, \cdots, o_j\}$, say $o_{l_0}$, whose distance to $\tilde{o}$ is no more than $(1+\frac{1+2\epsilon}{\sqrt{3-12\epsilon}})L$.\footnote{Note that we assume $\epsilon<\frac{1}{5}$ in Lemma \ref{lem-wsl}, so $(1+\frac{1+2\epsilon}{\sqrt{3-12\epsilon}})L$ is a finite real number.}
 \end{enumerate}}

\begin{proof}
We consider two cases:  (i) there are two subsets $P_{l_1}$ and $P_{l_2}$ of $P$ with size at least 
$\frac{1-\epsilon}{3j}|P|$, and (ii) no such pair of subsets exists.

For  case (i), by Markov's inequality, we know that there exist two points $q\in P_{l_1}$ and $q'\in P_{l_2}$ such that 
\begin{eqnarray}
||q-o||\leq \frac{3j}{1-\epsilon}\frac{1}{|P|}\sum_{p\in P}||p-o||; \label{for-wsl8}\\
||q'-o||\leq \frac{3j}{1-\epsilon}\frac{1}{|P|}\sum_{p\in P}||p-o||.\label{for-wsl9}
\end{eqnarray}
This, together with triangle inequality, indicates that both
$||o_{l_1}-o||$ and  $||o_{l_2}-o||$ are no more than $L+\frac{3j}{1-\epsilon}\frac{1}{|P|}\sum_{p\in P}||p-o||$. 
Since $\tilde{o}$ is the orthogonal projection of $o$ to $\mathcal{F}$,  we have $||o_{l_1}-\tilde{o}||\leq ||o_{l_1}-o||$ and $||o_{l_2}-\tilde{o}||\leq ||o_{l_2}-o||$. Thus, the first statement is true in this case.  

\begin{figure}[]
 \centering
  \includegraphics[height=1.2in]{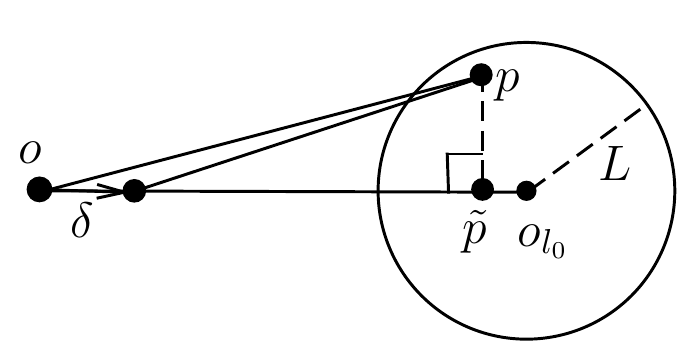}
        \caption{An illustration for  Claim 6.}
  \label{fig-claim7}
 \end{figure}

For case (ii), {\em i.e.,} no two subsets with size at least  $\frac{1-\epsilon}{3j}|P|$,
since $\sum^j_{l=1}|P_l| \ge (1-\epsilon)|P|$, by pigeonhole principle we know that there must exist one $P_{l_0}$, $1\leq l_0\leq j$, with size 
\begin{eqnarray}
|P_{l_0}|\geq (1-(j-1)\frac{1}{3j})(1-\epsilon)|P|\geq \frac{2}{3}(1-\epsilon)|P|. \label{for-wsl10}
\end{eqnarray}
Let $x=||o-o_{l_0}||$. We assume that $x>L$, since otherwise the second statement is automatically true.  

Now imagine moving $o$ slightly toward $o_{l_0}$ by a small distance $\delta$. See Figure~\ref{fig-claim7}. 
For any point $p\in P_{l_0}$,  let  $\tilde{p}$  be its orthogonal projection to the line $\overline{oo_{l_0}}$, and $a$ and $b$ be the distances $||o-\tilde{p}||$ and $||p-\tilde{p}||$, respectively.
Then, the distance between $p$ and $o$ is decreased by $\sqrt{a^2+b^2}-\sqrt{(a-\delta)^2+b^2}$. Also, we have
\begin{eqnarray}
\lim_{\delta\rightarrow 0}\frac{\sqrt{a^2+b^2}-\sqrt{(a-\delta)^2+b^2}}{\delta}&=&\lim_{\delta\rightarrow 0}\frac{2a-\delta}{\sqrt{a^2+b^2}+\sqrt{(a-\delta)^2+b^2}}\nonumber\\
&=&\frac{(a/b)}{\sqrt{(a/b)^2+1}}.\label{for-wsl11}
\end{eqnarray}
Since $p$ is inside ball $\mathcal{B}(o_{l_0}, L)$, we have $a/b\geq  (x-L)/L$. For any point $p\in P\setminus P_{l_{0}}$, the distance to $o$ is non-increased or increased by at most $\delta$. Thus, the average distance from the points in $P$ to $o$ is decreased by at least 
\begin{eqnarray}
\frac{2}{3}(1-\epsilon)\frac{((x-L)/L) \delta}{\sqrt{((x-L)/L)^2+1}}-(1-\frac{2}{3}(1-\epsilon))\delta. \label{for-wsl12}
\end{eqnarray}
Since the original position of $o$ is the  median point of $P$, the value of (\ref{for-wsl12}) should be non-positive. With  simple calculation, we have 
\begin{eqnarray}
(x-L)/L\leq \frac{1+2\epsilon}{\sqrt{3-12\epsilon}}\Longrightarrow x\leq (1+\frac{1+2\epsilon}{\sqrt{3-12\epsilon}})L. \label{for-wsl13}
\end{eqnarray}
By the same argument  in  case (i), we know that $||o_{l_0}-\tilde{o}||\leq ||o_{l_0}-o||$. This, together with (\ref{for-wsl13}),  implies that the second statement is true for case (ii).
This completes the proof for  this claim.
\qed
\end{proof}

With the above claims, we now prove Lemma \ref{lem-wsl}.

\begin{proof}[\textbf{of Lemma \ref{lem-wsl}}]
We build a grid in $\mathcal{F}$ as follows. First, draw a set of balls.
\begin{itemize}
\item For each $o_l$, $1\leq l\leq j$, draw a ball  (called {\em type-1 ball}) centered at $o_l$ and with radius $(1+\frac{1+2\epsilon}{\sqrt{3-12\epsilon}})L$.
\item For each pair of $o_{l}$ and $o_{l'}$, $1\leq l, l'\leq j$, draw a ball (called {\em type-2 ball}) centered at $o_{l}$ and with radius  $(1+\frac{1+2\epsilon}{\sqrt{3-12\epsilon}})(||o_{l}-o_{l'}||+L)$.
\end{itemize}

We claim that among the above balls, there must exist one ball that contains $\tilde{o}$. If there is only one  subset in $\{P_1, \cdots, P_j\}$ with size no smaller than $\frac{1-\epsilon}{3j}|P|$, it corresponds to the second case   in  Claim 6, and thus there exists a type-1 ball containing $\tilde{o}$. Now consider the case  that there are multiple subsets, say $\{P_{l_1}, \cdots, P_{l_t}\}$ for some $t\geq 2$,  all with size no smaller than $\frac{1-\epsilon}{3j}|P|$. 
Without loss of generality, assume that  $||o_{l_1}-o_{l_2}||=\max\{||o_{l_1}-o_{l_s}||\mid 1\leq s\leq t\}$. Then, we can view $\bigcup^t_{s=1}P_{l_s}$ as a big subset of $P$ bounded  by a ball centered at $o_{l_1}$ and with radius  $||o_{l_1}-o_{l_2}||+L$. By the same argument given in the proof   of  Claim 6 for (\ref{for-wsl10}), we know that $|\bigcup^t_{s=1}P_{l_s}|\geq \frac{2}{3}(1-\epsilon)|P|$. This also means that  we can reduce this case to the second case in  Claim 6, {\em i.e.,} replacing $P_{l_0}$, $o_{l_0}$ and $L$ by $|\bigcup^t_{s=1}P_{l_s}|$, $o_{l_1}$ and $||o_{l_1}-o_{l_2}||+L$ respectively. Thus,  there is a type-2 ball containing $\tilde{o}$.

Next, we discuss how to build the grids inside these balls. For type-1 balls with radius  $(1+\frac{1+2\epsilon}{\sqrt{3-12\epsilon}})L$, we build the grids inside them with grid length $\frac{\epsilon}{\sqrt{j}}L$. For type-2 balls with radius  $r_{l, l'}=(1+\frac{1+2\epsilon}{\sqrt{3-12\epsilon}})(||o_{l}-o_{l'}||+L)$ for some $l$ and $l'$, we build the grids inside them with grid length  
\begin{eqnarray}
\frac{1}{1+\frac{1+2\epsilon}{\sqrt{3-12\epsilon}}}\frac{(1-\epsilon)\epsilon}{6j\sqrt{j}}r_{l,l'}.\label{for-revise-1}
\end{eqnarray}

If $\tilde{o}$ is contained  in a type-1 ball, then there exists one grid point $\tau$ whose distance to $\tilde{o}$ is no more than $\epsilon L$. If $\tilde{o}$ is contained in a type-2 ball, such a distance is no more than 
\begin{eqnarray}
\frac{(1-\epsilon)\epsilon}{6j}(||o_{l}-o_{l'}||+L) \label{for-revise-2}
\end{eqnarray}
by (\ref{for-revise-1}). By the first statement in  Claim 6 and triangle inequality, we know that 
\begin{eqnarray}
||o_{l}-o_{l'}||\leq ||o_{l}-\tilde{o}||+||\tilde{o}-o_{l'}||\leq 2(L+\frac{3j}{1-\epsilon}\frac{1}{|P|}\sum_{p\in P}||p-o||). \label{for-revise-3}
\end{eqnarray} 
(\ref{for-revise-2}) and (\ref{for-revise-3}) imply that there exists one grid point $\tau$ whose distance to $\tilde{o}$ is no more than
\begin{eqnarray}
\epsilon \frac{1}{|P|}\sum_{p\in P}||p-o||+\frac{(1-\epsilon)\epsilon}{2j}L\leq \epsilon \frac{1}{|P|}\sum_{p\in P}||p-o||+\epsilon L. \label{for-wsl14}
\end{eqnarray}
Thus in both types of ball-containing,   
by  triangle inequality and  Claim 5, we have
\begin{eqnarray}
\frac{1}{|P|}\sum_{p\in P}||p-\tau||&\leq&\frac{1}{|P|}\sum_{p\in P}(||p-\tilde{o}||+||\tilde{o}-\tau||)\nonumber\\
&\leq&(\frac{1}{1-\epsilon}+\epsilon)\frac{1}{|P|}\sum_{p\in P}||p-o||+(1+\epsilon)L\nonumber\\
&\leq&(1+\frac{9}{4}\epsilon)\frac{1}{|P|}\sum_{p\in P}||p-o||+(1+\epsilon)L,\label{for-wsl15}
\end{eqnarray}
where the second inequality follows from the assumption that $\epsilon\leq \frac{1}{5}$.

As for the grid size, since we build the grids inside the balls in the $(j-1)$-dimensional flat $\mathcal{F}$, through simple calculation, we know that the grid size is $O(j^2(\frac{j\sqrt{j}}{\epsilon})^j)$. This completes the proof.
\qed
\end{proof}

\subsection{Peeling-and-Enclosing Algorithm for $k$-CMedian Using Weaker Simplex Lemma}
\label{sec-ptasmedian}

In this section,  we present a unified {\em Peeling-and-Enclosing} algorithm for generating  a set of candidate median points  for $k$-CMedian. Similar to the algorithm for $k$-CMeans, our algorithm iteratively determines the $k$ median points. 
At each iteration, it uses a set of peeling spheres and a simplex to search for an approximate median point. Since the simplex lemma no longer holds for $k$-CMedian, we use the weaker simplex lemma as a replacement.  Thus a number of changes are needed to accommodate the differences.


Before presenting our algorithm, we first introduce the following lemma proved by  Bad\u{o}iu {\em et al.} in \cite{BHI} for finding an approximate median point of a given point set.

%

\begin{theorem}[\cite{BHI}]
\label{the-1med}
Let $P$ be a normalized set of $n$ points in $\mathbb{R}^d$ space, $1>\epsilon>0$, and $R$ be a random sample of $O(1/\epsilon^3\log1/\epsilon)$ points from $P$. Then one can compute, in $O(d2^{O(1/\epsilon^4)}\log n)$ time, a point set  $S(P,R)$ of cardinality $O(2^{O(1/\epsilon^4)}\log n)$ , such that with constant probability (over the choices of $R$), there is a point $\tau\in S(P,R)$ such that $\sum_{p\in P}||\tau-p||\leq (1+\epsilon) \sum_{p\in P}||o-p||$, where $o$ is the optimal median point of $P$.
\end{theorem}

\begin{figure}[]
 \centering
  \includegraphics[height=1in]{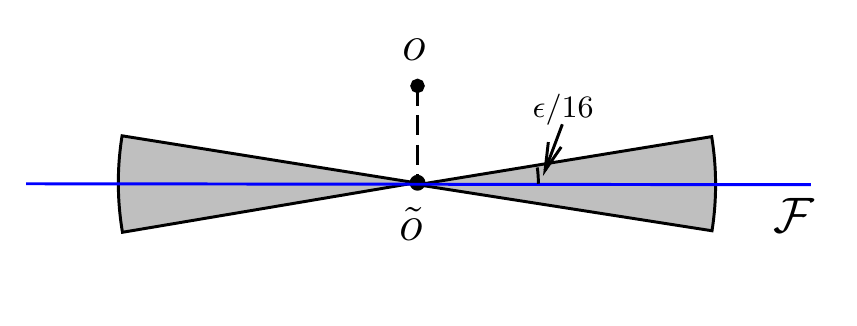}
        \caption{The gray area is $U$.}
  \label{fig-u}
 \end{figure}

\noindent\textbf{Sketch of the proof of Theorem \ref{the-1med}.} Since our algorithm uses some ideas in Theorem \ref{the-1med}, we  sketch its proof for completeness.  First, by Markov's inequality, we know that there exists one point, say $s_{1}$, from $R$ whose distance to $o$ is no more than $2 \frac{1}{|P|}\sum_{p\in P}||o-p||$ with certain probability. Then the sampling procedure can be viewed as an incremental process starting with $s_1$;  a flat $\mathcal{F}$ spanned by all previously obtained sample points is maintained;  at  each time that a new sample point is added,  $\mathcal{F}$ is updated. Let $\tilde{o}$ be  the projection of $o$ on $\mathcal{F}$, and 
\begin{eqnarray}
U=\{ p\in \mathbb{R}^d  \mid  \frac{\pi}{2}-\frac{\epsilon}{16}\leq \angle o\tilde{o}p\leq \frac{\pi}{2}+\frac{\epsilon}{16}\}. \label{for-cone}
\end{eqnarray}
See Figure \ref{fig-u}. It has been shown that this incremental sampling process stops before at most $O(1/\epsilon^3$ $\log1/\epsilon)$ points are taken, and one of the following two events happens with constant probability: (1) $\mathcal{F}$ is close enough to $o$, and (2) $|P\setminus U|$ is small enough. For either event, a grid can be built inside $\mathcal{F}$, and one of the grid points $\tau$ is the desired approximate median point.

Below we give an overview of our Peeling-and-Enclosing algorithm for $k$-CMedian. Let $P=\{p_1, \cdots, p_n\}$ be the set of $\mathbb{R}^{d}$ points in $k$-CMedian, and  $\mathcal{OPT}=\{Opt_1, \cdots,$ $Opt_k\}$ be the $k$ (unknown) optimal clusters with $m_{j}$ being the median point of cluster $Opt_j$ for $1\leq j\leq k$. Without loss of generality, we assume that $|Opt_1|\geq |Opt_2|\geq\cdots\geq |Opt_k |$. Denote  by $\mu_{opt}$ the optimal objective value, {\em i.e.,} $\mu_{opt}=\frac{1}{n}\sum^k_{j=1}\sum_{p\in Opt_j}||p-m_j||$.\\

\noindent\textbf{Algorithm overview:} We mainly focus on the differences with the $k$-CMeans algorithm. First, our algorithm uses Theorem \ref{the-1med} (instead of Lemma \ref{lem-dis}) to find an approximation $p_{v_{1}}$ for $m_1$. Then, it iteratively finds the approximate median points for $\{m_2, \cdots, m_k\}$ using the Peeling-and-Enclosing strategy. At the $(j+1)$-th iteration, 
it  has already obtained the approximate median points $p_{v_1}, \cdots, p_{v_j}$ for clusters $Opt_{1}, \cdots, Opt_{j}$, respectively.  To find the approximate median point $p_{v_{j+1}}$ for $Opt_{j+1}$, the algorithm draws $j$ peeling spheres $B_{j+1,1}, \cdots, B_{j+1,j}$ centered at $\{p_{v_1}, \cdots, p_{v_j}\}$, respectively, and considers the size of $\mathcal{A}=Opt_{j+1}\setminus (\bigcup^j_{l=1}B_{j+1,l})$. 
If $|\mathcal{A}|$ is small, it builds a flat (instead of  a simplex) spanned by $\{p_{v_1}, \cdots, p_{v_j}\}$, and finds  $p_{v_{j+1}}$  by using the weaker simplex lemma where the $j$ peeling spheres can be viewed as a partial partition on $Opt_{j+1}$. If $|\mathcal{A}|$ is large, it adopts a strategy similar to the one in Theorem \ref{the-1med} to find $p_{v_{j+1}}$:  start with the flat $\mathcal{F}$ spanned by $\{p_{v_1}, \cdots, p_{v_j}\}$, and grow  $\mathcal{F}$ by repeatedly adding a sample point in $\mathcal{A}$ to it. As it will be shown in Lemma \ref{lem-cone}, $\mathcal{F}$ will become close enough to $m_{j+1}$, and $p_{v_{j+1}}$ can be obtained by searching a grid (built in a way similar to Lemma \ref{lem-wsl}) in $\mathcal{F}$.
 By choosing a proper value ({\em i.e.,} $O(\epsilon)\mu_{opt}$) for $L$ in Lemma~\ref{lem-wsl} and Lemma~\ref{lem-cone}, we can achieve the desired  $(1+\epsilon)$-approximation. 
  As for the running time, although Theorem \ref{the-1med} introduces  an extra factor of $\log n$ for estimating the optimal cost of each $Opt_{j+1}$, our algorithm actually does not need it as such estimations have already been obtained during the Peeling-and-Enclosing step (see  Claim 2 in the proof of Lemma \ref{lem-induction}).
  Thus, the running time is still $O(n(\log n)^{k+1}d)$, which is the same as $k$-CMeans.

The algorithm is shown in Algorithm~\ref{alg-kcmedian}. The following lemma is needed to ensure the correctness of our algorithm.

\begin{lemma}
\label{lem-cone}
Let $\mathcal{F}$ be a flat in $\mathbb{R}^d$ containing $\{p_{v_1}, \cdots, p_{v_j}\}$ and having a distance to $m_{j+1}$ no more than $\frac{2}{|Opt_{j+1}|}\sum_{p\in Opt_{j+1}}||p-m_{j+1}||$. Assume that all  the peeling spheres $B_{j+1,1}, \cdots, B_{j+1,j}$ are centered at $\{p_{v_1}, \cdots, p_{v_j}\}$, respectively, and have a  radius $L\geq 0$. Then if $|Opt_{j+1}\setminus ((\bigcup^j_{w=1}B_{j+1,w})\bigcup U)|\leq \epsilon |Opt_{j+1}|$, we have
\begin{eqnarray}
&&\frac{1}{|Opt_{j+1}|}\sum_{p\in Opt_{j+1}}||p-\tilde{m}_{j+1}||\nonumber\\
&\leq& (1+2\epsilon)\frac{1}{|Opt_{j+1}|}\sum_{p\in Opt_{j+1}}||p-m_{j+1}||+L \label{for-cone1}
\end{eqnarray}
for any $0\leq \epsilon\leq 1$, where $\tilde{m}_{j+1}$ is the projection of $m_{j+1}$ on $\mathcal{F}$ and $U$ is defined in (\ref{for-cone}) ( after replacing $o$ and $\tilde{o}$ by $m_{j+1}$ and $\tilde{m}_{j+1}$, respectively).
\end{lemma}
\begin{proof}
To prove this lemma, we first compare it with Lemma \ref{lem-wsl}.  The main difference is that there is an extra part $U \cap Opt_{j+1}$ in $Opt_{j+1}$, where $Opt_{j+1}$ can be viewed as the point set $P$ in Lemma~\ref{lem-wsl}. Thus, $Opt_{j+1}$ can be viewed as having  three subsets, $(\bigcup^j_{w=1}B_{j+1,w}) \bigcap Opt_{j+1}$, $U \bigcap Opt_{j+1}$ and $Opt_{j+1}\setminus ((\bigcup^j_{w=1}B_{j+1,w})\bigcup U)$. 

For each point $p\in (\bigcup^j_{w=1}B_{j+1,w})\bigcap Opt_{j+1}$, similar to (\ref{for-wsl1}) in  Claim 5, we know that the cost  increases by at most $L$ if the median point moves from $m_{j+1}$ to $\tilde{m}_{j+1}$. Thus we have
\begin{eqnarray}
&&\sum_{p\in Opt_{j+1}\bigcap(\bigcup^j_{w=1}B_{j+1,w})}||p-\tilde{m}_{j+1}|| \nonumber\\
&\leq&\sum_{p\in Opt_{j+1}\bigcap(\bigcup^j_{w=1}B_{j+1,w})}(||p-m_{j+1}||+L). \label{for-cone2}
\end{eqnarray}

For the part $Opt_{j+1}\setminus ((\bigcup^j_{w=1}B_{j+1,w})\bigcup U)$,  by triangle inequality we have
\begin{eqnarray}
&&\sum_{p\in Opt_{j+1}\setminus ((\bigcup^j_{w=1}B_{j+1,w})\bigcup U)}||p-\tilde{m}_{j+1}||\nonumber\\
&\leq& \sum_{p\in Opt_{j+1}\setminus ((\bigcup^j_{w=1}B_{j+1,w})\bigcup U)}(||p-m_{j+1}||+||m_{j+1}-\tilde{m}_{j+1}||)\nonumber\\
 &\leq& \sum_{p\in Opt_{j+1}\setminus ((\bigcup^j_{w=1}B_{j+1,w})\bigcup U)}||p-m_{j+1}||+2\epsilon\sum_{p\in Opt_{j+1}}||p-m_{j+1}||, \label{for-cone3}
\end{eqnarray}
where the second inequality follows from the assumption that $\mathcal{F}$'s distance  to $m_{j+1}$ is no more than 
$\frac{2}{|Opt_{j+1}|}\sum_{p\in Opt_{j+1}}||p-m_{j+1}||$ and $$|Opt_{j+1}\setminus ((\bigcup^j_{w=1}B_{j+1,w})\bigcup U)|\leq \epsilon |Opt_{j+1}|.$$

For each point $p\in Opt_{j+1}\cap U$, recall that the angle $\angle m_{j+1}\tilde{m}_{j+1}p\in [\frac{\pi}{2}-\frac{\epsilon}{16}, \frac{\pi}{2}+\frac{\epsilon}{16}]$ in (\ref{for-cone}). In Theorem 3.2 of \cite{BHI}, it showed that $||p-\tilde{m}_{j+1}||\leq (1+\epsilon)||p-m_{j+1}||$. Therefore, 
\begin{eqnarray}
\sum_{p\in Opt_{j+1}\cap U}||p-\tilde{m}_{j+1}||\leq (1+\epsilon)\sum_{p\in Opt_{j+1}\cap U}||p-m_{j+1}||. \label{for-cone4}
\end{eqnarray}

Combining (\ref{for-cone2}), (\ref{for-cone3}) and (\ref{for-cone4}), we  obtain (\ref{for-cone1}).
\qed
\end{proof}

To complete the Peeling-and-Enclosing algorithm for $k$-CMedian, we also need 
an upper bound for the optimal objective value.  In Section \ref{sec-conmedian},  we will show how to obtain such an estimation. For this moment, we assume that the upper bound is available.   

\begin{algorithm}
   \caption{Peeling-and-Enclosing for $k$-CMedian}
   \label{alg-kcmedian}
\begin{algorithmic}
   \STATE {\bfseries Input:} $P=\{p_1, \cdots, p_n\}$ in $\mathbb{R}^d$, $k\geq 2$, a constant $\epsilon\in(0, \frac{1}{4k^2})$, and an upper bound $\Delta\in [\mu_{opt}, c\mu_{opt}]$ with $c\geq 1$.
   \STATE {\bfseries Output:} A set of $k$-tuple candidates for the $k$ constrained median points.
    \begin{enumerate}

\item For $i=0$ to $\lceil\log_{1+\epsilon}c\rceil$ do
\begin{enumerate}
\item Set $\mu=(1+\epsilon)^i\Delta/c$, and run Algorithm~\ref{alg-tree2}.

\item  Let $\mathcal{T}_i$ be the output tree.

\end{enumerate}

\item For each root-to-leaf path of  every $\mathcal{T}_i$, build a $k$-tuple candidate using the $k$ points associated with the path.

\end{enumerate}

\end{algorithmic}
\end{algorithm}

\begin{algorithm}
   \caption{Peeling-and-Enclosing-Tree \Rmnum{2}}
   \label{alg-tree2}
\begin{algorithmic}
   \STATE {\bfseries Input:} $P=\{p_1, \cdots, p_n\}$ in $\mathbb{R}^d$, $k\geq 2$, a constant $\epsilon\in(0, \frac{1}{4k^2})$, and $\mu>0$.
 \begin{enumerate}
\item Initialize $\mathcal{T}$ with a single root node $v$ associated with no point.
\item Recursively grow each node $v$ in the following way
\begin{enumerate}
\item If the height of $v$ is already $k$, then it is a leaf.
\item Otherwise, let $j$ be the height of $v$. Build the radius candidates set $\mathcal{R}=$ \\ $\cup^{\log n}_{t=0}\{\frac{1+l\frac{\epsilon}{2}}{2(1+\epsilon)}j2^{t}\epsilon\mu\mid 0\le l\le 4+\frac{2}{\epsilon}\}$. For each $r\in\mathcal{R}$, do
\begin{enumerate}
\item   Let $\{p_{v_1}, \cdots, p_{v_j}\}$ be the $j$ points associated with nodes on the root-to-$v$ path. 

\item For each $p_{v_l}$, $1\leq l\leq j$, construct a ball $B_{j+1,l}$ centered at $p_{v_l}$ and with radius $r$. 

\item Compute a flat spanned by $\{p_{v_1}, \cdots, p_{v_j}\}$, and build a grid inside it by Lemma~\ref{lem-wsl}.
\item Take a random sample from $P\setminus\cup^j_{l=1}B_{j+1,l}$ with size $s=\frac{k^3}{\epsilon^{11}}\ln\frac{k^2}{\epsilon^6}$, and compute the flat determined by these sample  points and $\{p_{v_1}, $ $\cdots, p_{v_j}\}$. Build a grid inside the flat by Theorem~\ref{the-1med}.


\item In total, there are $O(2^{poly(\frac{k}{\epsilon})})$ grid points inside these two flats. For each grid point, add one child to $v$, and associate it with the grid point.

\end{enumerate}

\end{enumerate}
\item Output $\mathcal{T}$.
\end{enumerate}
\end{algorithmic}
\end{algorithm}

Using the same idea for proving Theorem \ref{the-ptas}, we obtain the following theorem for $k$-CMedian.

 \begin{theorem}
\label{the-ptas2}
Let $P$ be a set of $n$ points in $\mathbb{R}^{d}$ and $k\in \mathbb{Z}^+$ be a fixed constant. In $O(2^{poly(\frac{k}{\epsilon})}n(\log n)^{k+1} d )$ time,  Algorithm~\ref{alg-kcmedian} outputs $O(2^{poly(\frac{k}{\epsilon})}$ $(\log n)^{k})$ $k$-tuple candidate median points. With constant probability, there exists one $k$-tuple candidate in the output which is able to induce a $\big(1+O(\epsilon)\big)$-approximation of $k$-CMedian (together with the solution for the corresponding Partition step).
\end{theorem}

\subsection{Upper Bound Estimation for $k$-CMedian}
\label{sec-conmedian}

In this section, we show how to obtain an upper bound of the optimal objective value of $k$-CMedian.


\begin{theorem}
\label{the-medconstant}
Let $P=\{p_1, \cdots, p_n\}$ be the input points of $k$-CMedian, and $\mathcal{C}$ be the set of $k$ median points of a $\lambda$-approximation of $k$-median on $P$ (without considering the constraint) for some constant $\lambda\geq 1$. Then the Cartesian product $[\mathcal{C}]^k$ contains at least one $k$-tuple which is able to induce a $(3\lambda+2)$-approximation of $k$-CMedian (together with the solution for the corresponding Partition step).
\end{theorem}
%

Let $\{c_1, \cdots, c_k\}$ be the $k$ median points in  $\mathcal{C}$, and $\omega$ be  the corresponding objective value of the $k$-median approximate solution on $P$.
Recall that $\{m_{1}, \cdots, m_{k}\}$ are the $k$ unknown optimal constrained median points  of $P$, and $\mathcal{OPT}=\{Opt_{1}, \cdots,$ $Opt_{k}\}$ are the corresponding $k$ optimal constrained clusters. To prove Theorem \ref{the-medconstant}, 
we  create a new instance of $k$-CMedian in the following way:  
for each point $p_i\in P$, move it to its nearest point, say $c_{t}$,  in $\{c_1, \cdots, c_k\}$; let $\tilde{p}_i$ denote the new $p_i$ 
 (note that $c_t$ and $\tilde{p}_i$ overlap with each other). Then the set $\tilde{P}=\{\tilde{p}_1, \cdots, \tilde{p}_n\}$ forms a new instance of $k$-CMedian. Let $\mu_{opt}$ and $\tilde{\mu}_{opt}$ be the optimal cost of $P$ and $\tilde{P}$ respectively, and $\mu_{opt}([\mathcal{C}]^k)$ 
 be the minimum cost of $P$ by restricting its $k$ constrained median points to being a $k$-tuple in $[\mathcal{C}]^k$.
 The following two lemmas are keys to proving Theorem \ref{the-medconstant}.

\begin{lemma}
\label{lem-mieq1}
$\tilde{\mu}_{opt} \leq \omega+ \mu_{opt}$.
\end{lemma}

\begin{proof}
%
For each $p_i\in Opt_l$, 
by triangle inequality we have
\begin{eqnarray}
||\tilde{p}_i-m_l||\leq ||\tilde{p}_i-p_i||+||p_i-m_l||.\label{for-mc1}
\end{eqnarray}
For both sides of (\ref{for-mc1}), taking the averages over $i$ and $l$, we get
\begin{eqnarray}
\frac{1}{n}\sum^k_{l=1}\sum_{p_i\in Opt_l}||\tilde{p}_i-m_l||  \leq  \frac{1}{n}\sum^n_{i=1} ||\tilde{p}_i-p_i||+\frac{1}{n}\sum^k_{l=1}\sum_{p_i\in Opt_l}||p_i-m_l||.\label{for-mc2}
\end{eqnarray}
Note that the left-hand side of (\ref{for-mc2}) is not smaller than $\tilde{\mu}_{opt}$, since $\tilde{\mu}_{opt}$ is the optimal object value of $k$-CMedian on $\tilde{P}$. For the right-hand side of (\ref{for-mc2}), the first term $\frac{1}{n} \sum^n_{i=1} ||\tilde{p}_i-p_i||=\omega$ (by the construction of $\tilde{P}$), and the second term $\frac{1}{n}\sum^k_{l=1}\sum_{p_i\in Opt_l}||p_i-m_l||=\mu_{opt}$. Thus, we  have $\tilde{\mu}_{opt} \leq \omega+ \mu_{opt}$. 
\qed
\end{proof}

\begin{lemma}
\label{lem-mieq2}
$\mu_{opt}([\mathcal{C}]^k)\leq \omega+2\tilde{\mu}_{opt}$.
\end{lemma}

\begin{proof}
Consider $k$-CMedian on $\tilde{P}$. Let $\{\tilde{m}_1, \cdots, \tilde{m}_k\}$ be the optimal constraint median points, and $\{\tilde{O}_1, \cdots, \tilde{O}_k\}$ be the corresponding optimal constraint clusters of $\tilde{P}$.  Let $\{\tilde{c}_1, \cdots, \tilde{c}_k\}$ be the $k$-tuple in $[\mathcal{C}]^k$ with $\tilde{c}_l$ being the nearest point in $\mathcal{C}$ to $\tilde{m}_l$. Thus, by an argument similar to the one for (\ref{for-mc1}), we have the following inequality, where $\tilde{p}_i$  is assumed to be clustered in $\tilde{O}_l$.
\begin{eqnarray}
||\tilde{p}_i-\tilde{c}_l|| &\leq &||\tilde{p}_i-\tilde{m}_l||+||\tilde{m}_l-\tilde{c}_l||\leq2||\tilde{p}_i-\tilde{m}_l||. \label{for-mc3}
\end{eqnarray}
In (\ref{for-mc3}), the last one follows from the facts that $\tilde{c}_l$ is the nearest point in $\mathcal{C}$ to $\tilde{m}_l$  and $\tilde{p}_i \in \mathcal{C}$, which implies $||\tilde{m}_l-\tilde{c}_l||\leq ||\tilde{m}_l-\tilde{p}_i||$.  For both sides of (\ref{for-mc3}), taking the averages over  $i$ and $l$, we have
\begin{eqnarray}
\frac{1}{n}\sum^k_{l=1}\sum_{\tilde{p}_i\in\tilde{O}_l}||\tilde{p}_i-\tilde{c}_l|| &\leq&  2\frac{1}{n}\sum^k_{l=1}\sum_{\tilde{p}_i\in\tilde{O}_l}||\tilde{p}_i-\tilde{m}_l||. \label{for-mc4}
\end{eqnarray}

%
Now, consider the following $k$-CMedian on $P$.  For each $p_i$, if  $ \tilde{p}_i\in\tilde{O}_l$, we cluster it to the corresponding median point $\tilde{c}_l$.
Then the objective value of the clustering is
\begin{eqnarray}
\frac{1}{n}\sum^k_{l=1}\sum_{\tilde{p}_i\in\tilde{O}_l}||p_i-\tilde{c}_l|| &\leq&\frac{1}{n} \sum^k_{l=1}\sum_{\tilde{p}_i\in\tilde{O}_l}(||p_i-\tilde{p}_i||+||\tilde{p}_i-\tilde{c}_l||)\nonumber\\
&\leq&\frac{1}{n} \sum^k_{l=1}\sum_{\tilde{p}_i\in\tilde{O}_l}||p_i-\tilde{p}_i||+2\frac{1}{n}\sum^k_{l=1}\sum_{\tilde{p}_i\in\tilde{O}_l}||\tilde{p}_i-\tilde{m}_l||.\label{for-mc7}
\end{eqnarray}
The left-hand side of (\ref{for-mc7}), {\em i.e.,} $\frac{1}{n}\sum^k_{l=1}\sum_{\tilde{p}_i\in\tilde{O}_l}||p_i-\tilde{c}_l||$, is no smaller than $\mu_{opt}([\mathcal{C}]^k)$ (by the definition), and the right-hand side of (\ref{for-mc7}) is equal to $\omega+2\tilde{\mu}_{opt}$. Thus, we  have $\mu_{opt}([\mathcal{C}]^k)\leq \omega+2\tilde{\mu}_{opt}$.
\qed
\end{proof}

 \begin{proof}[\textbf{of Theorem \ref{the-medconstant}}]
By Lemma \ref{lem-mieq1} and Lemma \ref{lem-mieq2}, we know that $\mu_{opt}([\mathcal{C}]^k)\leq 3\omega+2\mu_{opt}$. It is obvious that the optimal objective value of the $k$-median clustering is no larger than that of $k$-CMedian on the same set of points in $P$. This  implies that $\omega\leq \lambda \mu_{opt}$. Thus, we have 
\begin{eqnarray}
\mu_{opt}([\mathcal{C}]^k)\leq (3\lambda+2)\mu_{opt}.
\end{eqnarray}
The above inequality means that there exists one $k$-tuple in $[\mathcal{C}]^k$, which is able to induce a $(3\lambda+2)$-approximation.
\qed
\end{proof}


\subsection{Selection Algorithms for $k$-CMedian}
\label{sec-application2}

For each of the six constrained clustering problems studied in Section \ref{sec-application}, the same results (including the approximation ratio and time complexity) can be extended to the corresponding constrained $k$-median version with slight modification ({\em e.g.,} assigning the edge cost to be the Euclidean distance rather than squared Euclidean distance when computing the minimum cost circulation on the graph $G$). Thus, we only focus on the probabilistic clustering problem.\\

\noindent\textbf{Probabilistic $k$-Median Clustering ($k$-PMedian) \cite{CM08}.} Let $V=\{v_1, \cdots,$ $v_n\}$ be a set of nodes;  each node $v_{i}$ is associated with 
a point set $D_i=\{p^i_1, \cdots, p^i_h\}\subset \mathbb{R}^d$, where each $p^i_l$ has a probability $t^i_l\geq 0$ satisfying the condition   $\sum^h_{l=1}t^i_l\leq 1$.
Let $w_i=\sum^h_{l=1}t^i_l$ for $1\leq i\leq n$ and $W=\sum^n_{i=1}w_i$.
 $k$-PMedian is the problem of finding $k$ points $\{m_1, \cdots, m_k\}$ in $\mathbb{R}^d$ such that $\sum^n_{i=1}\min_{1\leq j\leq k}dist\{v_i,m_j\}$ is minimized, where $dist\{v_i,m_j\}=$ $\sum^h_{l=1}$ $t^i_l ||p^i_l-m_j||$.

Note that for the $k$-means version of probabilistic clustering, Cormode and McGregor\cite{CM08} have showed that it can be reduced to an ordinary $k$-means clustering problem 
after replacing each $D_i$ by its weighted mean point. However, this strategy
can only yield a $(3+\epsilon)$-approximation for the $k$-median version \cite{CM08,XX10}. We briefly sketch our idea for solving  $k$-PMedian  below.

Actually, $k$-PMedian is equivalent to the $k$-median clustering on the weighted point set $\bigcup^n_{i=1}D_i$ with the constraint that for each  $1\leq i\leq n$, all the points in $D_i$ should be clustered into the same cluster. Thus, we can use our  Peeling-and-Enclosing algorithm for $k$-CMedian in Section \ref{sec-ptasmedian} to generate a set of candidates for the constrained $k$ median points; the difference is that the points have weights, and thus in each sampling step we sample points with probabilities proportional to their weights. To accommodate such a difference, several minor  modifications need to be made to Lemma \ref{lem-wsl} and Lemma \ref{lem-cone}: all distances are changed to weighted distances, and the involved set sizes (such as $|P|$) are changed to $nh$. 

As for the running time of the Peeling-and-Enclosing algorithm, it still builds the trees with heights equal to $k$. But the number of children for each node is different.  Recall that in the proof of  Claim 2, in order to obtain an estimation for $\beta_j=\frac{|Opt_j|}{n}$, we need to try $O(\log n)$ times since $\frac{1}{n}\leq \beta_j\leq 1$; but for $k$-PMedian, the range of $\beta_j$ becomes $[\frac{w_{min}}{W}, 1]$ where $w_{min}=\min_{1\leq i\leq n}w_i$ (note that $W=\sum^n_{i=1}w_i\leq n$). Thus, the running time of Peeling-and-Enclosing algorithm becomes $O(nh (\log \frac{n}{w_{min}})^{k+1} d)$. Furthermore, for each $k$-tuple candidate, we perform the Partition step through assigning each $D_i$ to the $m_j$ with the smallest $dist\{v_i,m_j\}$. Obviously, the Partition step can be finished within linear time.  Thus we have the following theorem.


\begin{theorem}
\label{the-uncertain}
There exists an algorithm yielding a $(1+\epsilon)$-approximation for $k$-PMedian with constant probability, in $O(2^{poly(\frac{k}{\epsilon})}$ $nh$ $ (\log \frac{n}{w_{min}})^{k+1}$ $d)$ time,  where $w_{min}=\min_{1\leq i\leq n}w_i$. 
\end{theorem}

\section{Future Work}

Following this work, some interesting problems deserve to be further studied in the future. For example, we reduce the partition step to the minimum cost circulation problem for several constrained clustering problems in Section~\ref{sec-application}; however, since the goal is to find an approximate solution, one may consider using the geometric information to solve the Partition step approximately. In Euclidean space, several techniques have been developed for solving approximate matching problems efficiently \cite{AON,SA122}. But it is still not clear whether such techniques  can be extended to solve the constrained matching problems (such as the $r$-gather or $l$-diversity) considered in this paper, especially in high dimensional space. We leave it as an open problem for future work.




\section{Appendix}
\label{sec-detail}

\subsection{Proof of Lemma \ref{lem-shift}}
\label{sec-detail_shift}

Similar to Lemma \ref{lem-simplex}, we prove this lemma by mathematical induction on $j$. 

\textbf{Base case:} For $j=1$, since $o_1=o$, we just need to let $\tau=o'_1$. Then, we have 
\begin{eqnarray}
||\tau-o||=||o'_1-o||=||o'_1-o_1||\leq L\leq\sqrt{\epsilon}\delta+(1+\epsilon)L. 
\end{eqnarray}
Thus, the base case holds.

\textbf{Induction step:} Assume that the lemma holds for any $j\leq j_0$ for some $j_{0} \ge 1$ ({\em i.e.}, the induction hypothesis). Now we consider the case of $j=j_0+1$. Similar to the proof of Lemma~\ref{lem-simplex}, we assume that $\frac{|Q_l |}{|Q|}\geq \frac{\epsilon}{4j}$ for each $1\leq l\leq j$. Otherwise, through a similar idea from Lemma~\ref{lem-simplex}, it can be reduced to the case with smaller $j$, and solved by the induction hypothesis.  Hence, in the following discussion, we assume that $\frac{|Q_l |}{|Q|}\geq \frac{\epsilon}{4j}$ for each $1\leq l\leq j$. 

First, we know that $o=\sum^j_{l=1}\frac{|Q_l |}{|Q|} o_l$. Let $o'=\sum^j_{l=1}\frac{|Q_l |}{|Q|} o'_l$. Then, we have 
\begin{eqnarray}
||o-o'||=||\sum^j_{l=1}\frac{|Q_l |}{|Q|} o_l-\sum^j_{l=1}\frac{|Q_l |}{|Q|} o'_l||\leq \sum^j_{l=1}\frac{|Q_l |}{|Q|}||o_l-o'_l||\leq L. \label{for-8}
\end{eqnarray}
Thus, if we can find a grid point $\tau$ having $||\tau-o'||\leq \sqrt{\epsilon}\delta+\epsilon L$, by inequality (\ref{for-8}), we will have $||\tau-o||\leq||\tau-o'||+||o'-o||\leq\sqrt{\epsilon}\delta+(1+\epsilon)L$. So  we only need to find a grid point close enough to $o'$. 

To find such a $\tau$, we consider the distance from $o'_{l}$ to $o'$ for any $1\leq l\leq j$. We have
\begin{eqnarray}
||o'_l-o'||\leq ||o'_l-o_l||+||o_l-o||+||o-o'||\leq 2\sqrt{\frac{j}{\epsilon}}\delta+2L, \label{for-9}
\end{eqnarray}
where the first inequality follows from triangle inequality, and the second inequality follows from the facts that $||o'_l-o_l||$ and $||o-o'||$ are both bounded by $L$, and $||o_l-o||\leq  2\sqrt{\frac{j}{\epsilon}}\delta$ (by Lemma \ref{lem-close}).

This implies that we can use a similar idea in Lemma \ref{lem-simplex} to construct a ball $\mathcal{B}$ centered at any  $o'_{l_0}$ and with radius $r=\max_{1\leq l\leq j}\{||o'_l-o'_{l_0}||\}$. Also, the simplex $\mathcal{V}'$ is inside $\mathcal{B}$. Note that 
 \begin{eqnarray}
 ||o'_l-o'_{l_0}||\leq ||o'_l-o'||+||o'-o'_{l_0}||\leq 4\sqrt{\frac{j}{\epsilon}}\delta+4L\label{for-simplex2_radii}
 \end{eqnarray}
 by (\ref{for-9}), which implies $r\leq 4\sqrt{\frac{j}{\epsilon}}\delta+4L$. Similar to Lemma \ref{lem-simplex}, we can build a grid inside $\mathcal{B}$ with grid length $\frac{\epsilon r}{4j}$, and the number of grid points is $O((8j/\epsilon)^j)$. Moreover, $o'$ must lie inside $\mathcal{V}'$ by the definition. In this grid, we can find a grid point $\tau$ such that $||\tau-o'||\leq \frac{\epsilon}{4\sqrt{j}}r\leq \sqrt{\epsilon}\delta+\epsilon L$. Thus, $||\tau-o||\leq||\tau-o'||+||o'-o||\leq\sqrt{\epsilon}\delta+(1+\epsilon)L$, and the induction step, as well as the lemma, holds.

%
%
%
%
%

\subsection{Proof of  Claim 2 for Lemma \ref{lem-induction}}
\label{sec-detail_claim3}

Since $1\geq \beta_j\geq \frac{1}{n}$, there is one integer $t$ between $1$ and $\log n$, such that $2^{t-1}\leq\frac{1}{\beta_j}\leq 2^t$. Thus $ 2^{t/2-1}\sqrt{\epsilon}\delta_{opt}\leq\sqrt{\frac{\epsilon}{\beta_j}}\delta_{opt}\leq 2^{t/2}\sqrt{\epsilon}\delta_{opt}$. Together with $\delta\in [\delta_{opt}, (1+\epsilon)\delta_{opt}]$, we have
\begin{eqnarray}
2^{t/2-1}\sqrt{\epsilon}\frac{\delta}{1+\epsilon}\leq\sqrt{\frac{\epsilon}{\beta_j}}\delta_{opt}\leq 2^{t/2}\sqrt{\epsilon}\delta .
\end{eqnarray}
 Thus if  setting $\hat{r}_j=2^{t/2}\sqrt{\epsilon}\delta$, we have 
\begin{eqnarray}
\sqrt{\frac{\epsilon}{\beta_j}}\delta_{opt}\leq \hat{r}_j\leq 2(1+\epsilon)\sqrt{\frac{\epsilon}{\beta_j}}\delta_{opt}.\label{for-revise-appendix1} 
 \end{eqnarray}
We consider the interval $\mathcal{I}=[\frac{j}{2(1+\epsilon)}\hat{r}_j, j\hat{r}_j]$. (\ref{for-revise-appendix1}) ensures that $j\sqrt{\frac{\epsilon}{\beta_j}}\delta_{opt}\in\mathcal{I}$. Also, we build a grid in the interval with grid length $\frac{\epsilon}{2}\frac{1}{2(1+\epsilon)}j\hat{r}_j$, {\em i.e.,} $\mathcal{R}_j=\{\frac{1+l\frac{\epsilon}{2}}{2(1+\epsilon)}j\hat{r}_j\mid 0\le l\le 4+\frac{2}{\epsilon}\}$. Moreover, the grid length $\frac{\epsilon}{2}\frac{1}{2(1+\epsilon)}j\hat{r}_j\leq \frac{\epsilon}{2}j\sqrt{\frac{\epsilon}{\beta_j}}\delta_{opt}$, which implies that 
%
%
%
%
%
%
%
there exists $r_j\in \mathcal{R}_j$ such that
\begin{eqnarray}
j\sqrt{\frac{\epsilon}{\beta_j}}\delta_{opt}\leq r_j\leq (1+\frac{\epsilon}{2})j\sqrt{\frac{\epsilon}{\beta_j}}\delta_{opt}.
\end{eqnarray}
Note that $\mathcal{R}_j\subset \mathcal{R}$, where $\mathcal{R}=\cup^{\log n}_{t=0}\{\frac{1+l\frac{\epsilon}{2}}{2(1+\epsilon)}j2^{t/2}\sqrt{\epsilon}\delta\mid 0\le l\le 4+\frac{2}{\epsilon}\}$. Thus, the Claim is true.
%

\subsection{Proof of  Claim 3 for Lemma \ref{lem-induction}}
\label{sec-detail_claim4}
Note that $\delta^2_{opt}=\sum^k_{j=1}\beta_j\delta^2_j$, and $\beta_j\le \beta_l$  for each $1\leq l\leq j-1$. Thus, we have $\delta_{l}\leq \sqrt{\frac{1}{\beta_l}}\delta_{opt}\le \sqrt{\frac{1}{\beta_j}}\delta_{opt}$. Together with $j\sqrt{\frac{\epsilon}{\beta_j}}\delta_{opt}\leq r_j$ (Claim 2) and $||p_{v_l}-m_l ||\leq \epsilon\delta_l+(1+\epsilon)l\sqrt{\frac{\epsilon}{\beta_l}}\delta_{opt} $ (by the induction hypothesis), we have 
\begin{eqnarray}
r_j-||p_{v_l}-m_l ||&\ge& j\sqrt{\frac{\epsilon}{\beta_j}}\delta_{opt}-(\epsilon\delta_l+(1+\epsilon)(j-1)\sqrt{\frac{\epsilon}{\beta_l}}\delta_{opt} )\nonumber\\
&\geq& (1-(j-1)\epsilon)\sqrt{\frac{\epsilon}{\beta_j}}\delta_{opt}-\epsilon\delta_l\nonumber\\
&\geq& (1-(j-1)\epsilon-\sqrt{\epsilon})\sqrt{\frac{\epsilon}{\beta_j}}\delta_{opt}. \label{for-revise-appendix2}
\end{eqnarray}
Since $\epsilon\in(0, \frac{1}{4k^2})$ in the input of Algorithm~\ref{alg-kcmeans}, we know $r_j-||p_{v_l}-m_l ||>0$. That is, $m_l$ is covered by the ball $B_{j,l}$.

%

For each $1\leq l\leq j-1$, we have $|Opt_l \setminus (\bigcup^{j-1}_{w=1}B_{j,w})|\leq|Opt_l \setminus B_{j,l} | $. For any $p\in Opt_l \setminus B_{j,l}$, $||p-m_l||\geq r_j-||p_{v_l}-m_l ||$. 
By Markov's inequality, we have
\begin{eqnarray}
|Opt_l \setminus B_{j,l} | \leq \frac{\delta^2_l}{(r_j-||p_{v_l}-m_l ||)^2}|Opt_l |.
\end{eqnarray}



Together with (\ref{for-revise-appendix2}), we have
\begin{eqnarray}
|Opt_l \setminus B_{j,l} | &\leq& \frac{\delta^2_l}{(1-(j-1)\epsilon-\sqrt{\epsilon})^2\frac{\epsilon}{\beta_j}\delta^2_{opt}}|Opt_l |\nonumber\\
&\leq& \frac{\delta^2_l}{(1-(j-1)\epsilon-\sqrt{\epsilon})^2\frac{\epsilon}{\beta_j}\beta_l \delta^2_l}|Opt_l | \nonumber\\
&=& \frac{\beta_j}{(1-(j-1)\epsilon-\sqrt{\epsilon})^2\epsilon\beta_l}|Opt_l |\nonumber\\
&=&\frac{\beta_j n}{(1-(j-1)\epsilon-\sqrt{\epsilon})^2\epsilon}\leq\frac{\beta_j n}{(1-j\sqrt{\epsilon})^2\epsilon},  \label{for-revise-appendix3}
\end{eqnarray}
%
where the second inequality follows from the fact that $\beta_l \delta^2_l\leq \delta^2_{opt}$, and the fourth equation follows from that $\frac{|Opt_l|}{\beta_l}=n$. Again, $\epsilon\in(0, \frac{1}{4k^2})$ implies that $ \frac{\beta_j n}{(1-j\sqrt{\epsilon})^2\epsilon}\leq \frac{4\beta_j n}{\epsilon}$. 
Thus, in total, we have
\begin{eqnarray}
|Opt_l \setminus B_{j,l} | &\leq& \frac{4\beta_j n}{\epsilon}.
\end{eqnarray}
Hence, the Claim is true.





\end{document}